%% file: main.tex
\documentclass[11pt]{article}
\input{header}

\input{macros}

\title{
Parallel Reachability and Shortest Paths on Non-sparse Digraphs: Near-linear Work and Sub-square-root Depth
}

\author{Vikrant Ashvinkumar \thanks{Rutgers.} \and Aaron Bernstein \thanks{New York University. Supported by Sloan Fellowship, Google Research Fellowship,  NSF Grant 1942010, and Charles S. Baylis endowment at NYU.} \and Maximilian Probst Gutenberg \thanks{ETH Zurich. The research leading to these results has received funding from grant no. 200021 204787 of the Swiss National Science Foundation.} \and Thatchaphol Saranurak \thanks{University of Michigan. Supported by NSF Grant CCF-2238138 and a Sloan Fellowship.}}
\date{}

\begin{document}
\maketitle
\begin{abstract}
We present parallel algorithms for computing single-source reachability and shortest paths on directed $n$-vertex $m$-edge graphs using near-linear $\tilde{O}(m)$ work and $o(\sqrt{n})$ depth whenever $m\ge n^{1+o(1)}$.
At the extreme of $m=\Omega(n^{2})$, our reachability and shortest path algorithms have depth only $n^{0.136}$ and $n^{0.25+o(1)}$, respectively. The state-of-the-art parallel algorithms with near-linear work for both problems \cite{jambulapati2019parallel,cao2020efficient,rozhovn2023parallel,cao2023parallel,brand2025parallel} require $\Omega(\sqrt{n})$ depth in all density regimes.
\end{abstract}

\pagenumbering{gobble}
\clearpage
\tableofcontents
\clearpage
\pagenumbering{arabic}
\setcounter{page}{1}
\input{intro}
\input{preliminaries}
\input{overview}
\input{shortcut-set}
\input{reach}
\input{hopset}
\input{sssp}
\input{open-problems}
\section*{Acknowledgments}
We would like to thank Shang-En Huang, who declined co-authorship despite being involved in our discussions.

\bibliographystyle{alpha}
\bibliography{references}

\appendix
\input{sketch-key-lemma}

\end{document}

%% file: header.tex
\usepackage[margin=1in]{geometry}
\usepackage[utf8]{inputenc} %
\usepackage[svgnames]{xcolor}
\usepackage{graphicx,xspace} %
\usepackage{pifont}
\usepackage{amsthm,amsmath,amssymb,mathtools}
\usepackage{tcolorbox}
\usepackage{subfig}
\usepackage{enumerate}
\usepackage[abs]{overpic}
\tcbuselibrary{skins,breakable}
\tcbset{enhanced jigsaw}

\usepackage{thmtools} 
\usepackage{thm-restate}

\usepackage[linktocpage=true,
	pagebackref=true,
        colorlinks,linkcolor=red,citecolor=olive,
	bookmarks,bookmarksopen,bookmarksnumbered]
	{hyperref}
\usepackage[noabbrev,nameinlink]{cleveref}

\usepackage{ifthen}
\usepackage{todonotes}

\newtheorem{lemma}{Lemma}[section]

\newtheorem{proposition}[lemma]{Proposition}
\newtheorem{observation}[lemma]{Observation}

\newtheorem{claim}[lemma]{Claim}

\newtheorem{remark}[lemma]{Remark}

\crefname{lemma}{Lemma}{Lemmas}
\crefname{appendix}{Appendix}{Appendices}
\crefname{proposition}{Proposition}{Propositions}
\crefname{observation}{Observation}{Observations}
\crefname{claim}{Claim}{Claims}
\crefname{figure}{Figure}{Figures}

%% file: macros.tex
\newcommand{\eps}{\ensuremath{\varepsilon}}
\newcommand{\Ot}{\ensuremath{\widetilde{O}}}

\newcommand{\bracket}[1]{\left[#1\right]}
\newcommand{\paren}[1]{\left(#1\right)}
\newcommand{\card}[1]{\left\vert{#1}\right\vert}
\newcommand{\set}[1]{\left\{ #1 \right\}}

\newcommand{\IN}{\ensuremath{\mathbb{N}}}
\newcommand{\cA}{\ensuremath{{\cal A}}}

\newcommand{\prob}[1]{\Pr\bracket{#1}}
\newcommand{\expect}[1]{\Exp\bracket{#1}}
\newcommand{\expectt}[2]{\Exp_{#1}\bracket{#2}}

\DeclareMathOperator*{\Exp}{\ensuremath{{\mathbb{E}}}}
\newcommand{\algname}[1]{\underline{\textbf{\textsc{#1}}}\\}
\newenvironment{tbox}{\begin{tcolorbox}[
		enlarge top by=5pt,
		enlarge bottom by=5pt,
		 breakable,
		 boxsep=2pt,
                  left=5pt,
                  right=7pt,
                  top=10pt,
                  arc=0pt,
                  boxrule=1pt,toprule=1pt,
                  colback=white
                  ]%
	}
{\end{tcolorbox}}
\newenvironment{ttbox}{\begin{tcolorbox}[
		enlarge top by=5pt,
		enlarge bottom by=5pt,
		 breakable,
		 boxsep=2pt,
                  left=5pt,
                  right=7pt,
                  top=10pt,
                  arc=0pt,
                  boxrule=1pt,toprule=1pt,
                  colback=white,
                  colframe=lightgray
                  ]%
	}
{\end{tcolorbox}}

\newcommand{\cE}{\mathcal{E}}
\newcommand{\dist}{\mathrm{dist}}
\newcommand{\bdist}[1]{\mathrm{dist}^{\paren{#1}}}
\newcommand{\length}{\mathrm{length}}
\newcommand{\JLS}{\textsc{JLS}\xspace}
\newcommand{\CFR}{\textsc{CFR}\xspace}
\newcommand{\TC}{\textsc{TC}\xspace}

\newcommand{\TDijkstra}{\textsc{TruncSSSP}\xspace}

\newcommand{\anc}[2]{R^{-}\hspace{-0.2em}\paren{#1, #2}}
\newcommand{\danc}[3]{R^{-}_{#3}\hspace{-0.2em}\paren{#1, #2}}
\newcommand{\vanc}[1]{{#1}^{\mathrm{Anc}}}
\newcommand{\desc}[2]{R^{+}\hspace{-0.2em}\paren{#1, #2}}
\newcommand{\ddesc}[3]{R^{+}_{#3}\hspace{-0.2em}\paren{#1, #2}}
\newcommand{\vdesc}[1]{{#1}^{\mathrm{Des}}}
\newcommand{\relevant}[2]{\mathrm{R}\hspace{-0.1em}\paren{#1, #2}}
\newcommand{\drelevant}[3]{\mathrm{R}_{#3}\hspace{-0.1em}\paren{#1, #2}}
\newcommand{\children}[1]{\mathrm{children}\paren{#1}}
\newcommand{\cprunei}{DarkMagenta}

\newcommand{\cpruneii}{FireBrick}

\newcommand{\prunei}{\textcolor{\cprunei}{\TC-Pruning}\xspace}

\newcommand{\pruneii}{\textcolor{\cpruneii}{\TDijkstra-Pruning}\xspace}

\newcommand{\omeganc}{\omega\xspace}
\renewcommand{\rho}{\varrho}

%% file: intro.tex
\section{Introduction}
\label{sec:intro}

\newboolean{restate}
\setboolean{restate}{false}

Single-source reachability and shortest paths (SSSP) are some of the most fundamental algorithmic problems on digraphs.
Given a digraph $G = (V,E)$ on $n$ vertices and $m = \Omega(n)$ edges, and a source $s \in V$, reachability asks for the set of all vertices that $s$ can reach, and shortest paths for the distances from $s$ to $v$ for all $v \in V$.
While the solutions to these problems are fairly well understood in the sequential setting --- reachability is solvable in $O(m)$ time using a breadth-first search (BFS) and shortest paths in $\Ot(m)$ time\footnote{Here and throughout, we use the $\Ot(\cdot)$ notation to suppress polylogarithmic factors.} using Dijkstra's algorithm --- we are much more limited in our understanding of these problems in other computational models.
We focus on these problems under the parallel setting in this paper.

The \emph{work} $W(n)$ of a parallel algorithm is the number of primitive operations it performs on an input of size $n$, which can be interpreted as the running time of the algorithm when sequentialized, i.e., the running time with just one processor.
Ideally, $W(n)$ asymptotically equals the best sequential running time in which case we call the algorithm work-efficient.
For example, a work-efficient algorithm for reachability would have $W(\card{E(G)}) = O(m)$, matching the running time of BFS.
We call a parallel algorithm nearly work-efficient if $W(n)$ matches the best sequential running time up to polylogarithmic factors.
We interchangeably use the terms \emph{depth} or \emph{span} $D(n)$ of a parallel algorithm to refer to the length of the longest chain of sequential dependencies in the algorithm, which can be interpreted as the fastest parallel time the algorithm could possibly run in, i.e., the running time when given an unconstrained number of processors.

In this paper, we give nearly work-efficient algorithms for both single-source reachability and SSSP that have $o(\sqrt{n})$ depth on non-sparse digraphs.
Prior to our work, the state-of-the-art nearly work-efficient algorithms for both reachability and SSSP had depth at least $\Omega(\sqrt{n})$ for all graph densities, achieved by \cite{jambulapati2019parallel,brand2025parallel} and \cite{rozhovn2023parallel,cao2023parallel,brand2025parallel} respectively.

\subsection{Prior Work and Background}

\paragraph{Background on Parallel Reachability. }
The two most basic algorithms for parallel reachability are (i) using parallel matrix multiplication to compute the transitive closure of $G$ which, while highly parallel with $\textrm{polylog}(n)$ depth, is not remotely work-efficient; and (ii) using parallel BFS which, while work-efficient, unfortunately has span proportional to the depth of the BFS-tree of $G$, which is $\Omega(n)$ on worst-case inputs.
Until somewhat recently, no nearly work-efficient algorithm with sublinear depth had been known.
The breakthrough of Fineman~\cite{fineman2018nearly} gave the first nearly work-efficient algorithm with sublinear depth $\Ot(n^{2/3})$.
This was then built upon by Jambulapati, Liu, and Sidford~\cite{jambulapati2019parallel} to give a nearly work-efficient algorithm with $n^{1/2 + o(1)}$ depth, which remains the best known upper bound. This algorithm can also be generalized to give a work-span tradeoff \cite{cao2021brief}.

Both \cite{fineman2018nearly,jambulapati2019parallel} utilize a framework based on \emph{shortcut sets}, which dates back to \cite{ullman1991high}.
A $\beta$-shortcut set of $G$ is a set $H \subseteq V \times V$ such that the reachability relations of $G$ and $G \cup H$ are the same, and for any $u,v \in V$ such that $u$ reaches $v$, there is a path from $u$ to $v$ using at most $\beta$ hops in $G \cup H$.
Given a $\beta$-shortcut set $H$, one may compute single-source reachability in $O(m + \card{H})$ work and $O(\beta)$ depth by computing a parallel BFS on $G \cup H$ from source $s$.
There is a folklore construction of an $O(\sqrt{n})$-shortcut set $H$ with $\card{H} = \Ot(n)$: randomly sample $\Ot(\sqrt{n})$ vertices and add all shortcut edges $uv$ where $u$ and $v$ are sampled vertices such that $u$ can reach $v$ in $G$.
This was (at the time of \cite{fineman2018nearly,jambulapati2019parallel}) the best known existential bound, but no efficient algorithmic constructions meeting these bounds were known, even in the \emph{sequential} setting.
The breakthrough of \cite{fineman2018nearly} was a nearly work-efficient construction of an $\Ot(n^{2/3})$-shortcut set with linear size.
\cite{jambulapati2019parallel} then gave the first nearly work-efficient construction of a shortcut set that essentially matches the folklore bound.
The key technical contribution to both results was the first near-linear time sequential construction of said shortcut sets, which they then showed how to parallelilize. 

Both \cite{fineman2018nearly} and \cite{jambulapati2019parallel} construct a shortcut set $H$ with only $\Ot(n)$ edges. But for many problems, and for parallel reachability in particular, a shortcut set with $\Ot(m)$ edges would do just as well.
For non-sparse graphs, allowing $H$ to contain more edges is conceptually advantageous because the folklore construction can be tuned to achieve a $O(n/\sqrt{m})$-shortcut set $H$ with $\card{H} = \Ot(m)$, by simply sampling more vertices.
Unfortunately, the constructions of \cite{fineman2018nearly,jambulapati2019parallel} are not so readily tunable in this way, so we do not have any efficient algorithmic constructions of these denser hopsets, even in the sequential settings. In this paper, we give a near-linear time sequential construction of an $\Ot(m)$ size $\beta$-shortcut set where $\beta$ scales with the density of the base graph; for any $m$ polynomially larger than $n$, we achieve $\beta$ polynomially better than $\sqrt{n}$.
We can further improve the tradeoff with fast matrix multiplication: at the extreme case, if $\omega = 2$, our construction essentially matches the bounds of the tuned folklore construction, generalizing \cite{jambulapati2019parallel}.
We then show that all our sequential bounds can be easily parallelized, leading to 
nearly work-efficient parallel reachability algorithms with significantly lower depth in non-sparse graphs.

\paragraph{Background on Parallel Shortest Paths. }
Rozho\v{n} \textcircled{r}\footnote{A randomized author ordering generated in the cited work is delimited by \textcircled{r} separators.} Haeupler \textcircled{r} Martinsson \textcircled{r} Grunau \textcircled{r} Zuzic~\cite{rozhovn2023parallel} gave a blackbox reduction from SSSP to single-source $(1+\eps)$-approximate shortest paths, so we focus on the latter.

A similar framework to computing shortcut sets followed by parallel BFS is viable for approximate shortest paths.
Instead of BFS, one uses an algorithm of Klein and Subramanian \cite{klein1997randomized} for finding (exact) shortest $\beta$ hop paths in $\Ot(m)$ work and $\Ot(\beta)$ span.
And instead of a $\beta$-shortcut set, one constructs a so-called $(\beta,\eps)$-\emph{hopset} of $G$ which is a set $H \subseteq V \times V$ such that $\dist_G(u,v) \le \bdist{\beta}_{G \cup H}(u,v) \le (1+\eps)\dist_G(u,v)$ for all $u,v \in V$, where $\bdist{\beta}$ is the length of the shortest path using at most $\beta$ hops.
The folklore shortcut set is easily adapted to showing the existence of $(O(\sqrt{n}),\eps)$-hopsets $H$ with $\card{H} =\Ot(n)$.
Analogously, Cao, Fineman, and Russell~\cite{cao2020efficient} adapted the shortcut set construction of \cite{jambulapati2019parallel} to give a near-linear time construction of an $(n^{1/2+o(1)},\eps)$-hopset $H$ with $\card{H} =\Ot(n)$, essentially meeting this folklore bound.
A work-span tradeoff was then later given in \cite{cao2021brief}.

Just as with shortcut sets, the requirement that $\card{H} = \Ot(n)$ is sometimes too stringent.
The folklore hopset can in the same way be tuned to give an $(O(n/\sqrt{m}),\eps)$-hopset $H$ with $\card{H} = \Ot(m)$ by sampling more vertices, but the hopset of \cite{cao2020efficient} does not immediately generalize in this way.
In this paper, we give a near-linear time sequential construction of an $\Ot(m)$ size $(\beta,\eps)$-hopset where $\beta$ scales with the density of the base graph.
This is then parallelized and plugged into the abovementioned framework to get faster nearly work-efficient parallel SSSP algorithms.

\subsection{Our Results}

See \Cref{fig:plot} for a quick comparison of our main results with the strongest known bounds prior to this work.
Note that the parallel SSSP algorithm is purely combinatorial; for reachability, we get a better tradeoff using fast matrix multiplication.
\clearpage

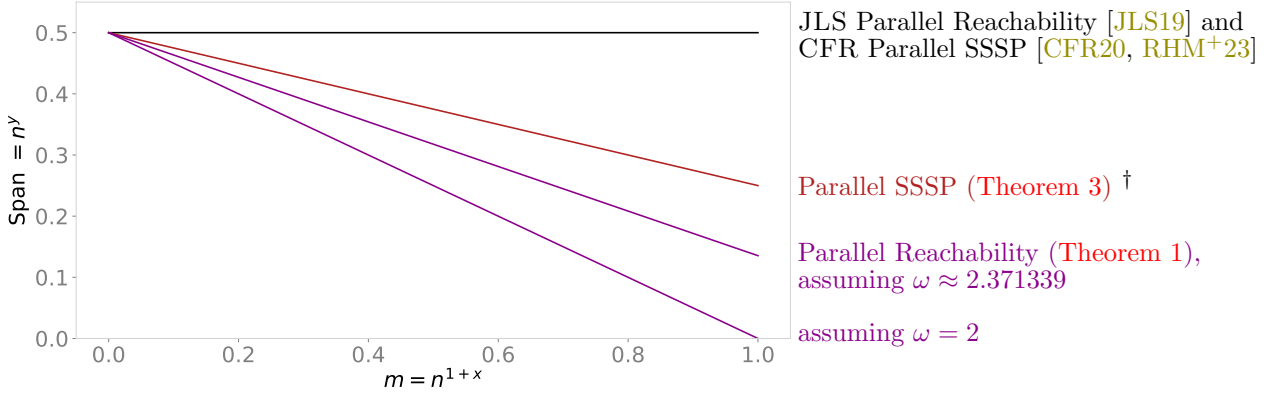
\begin{figure}[h]
    \begin{overpic}[scale=0.4]{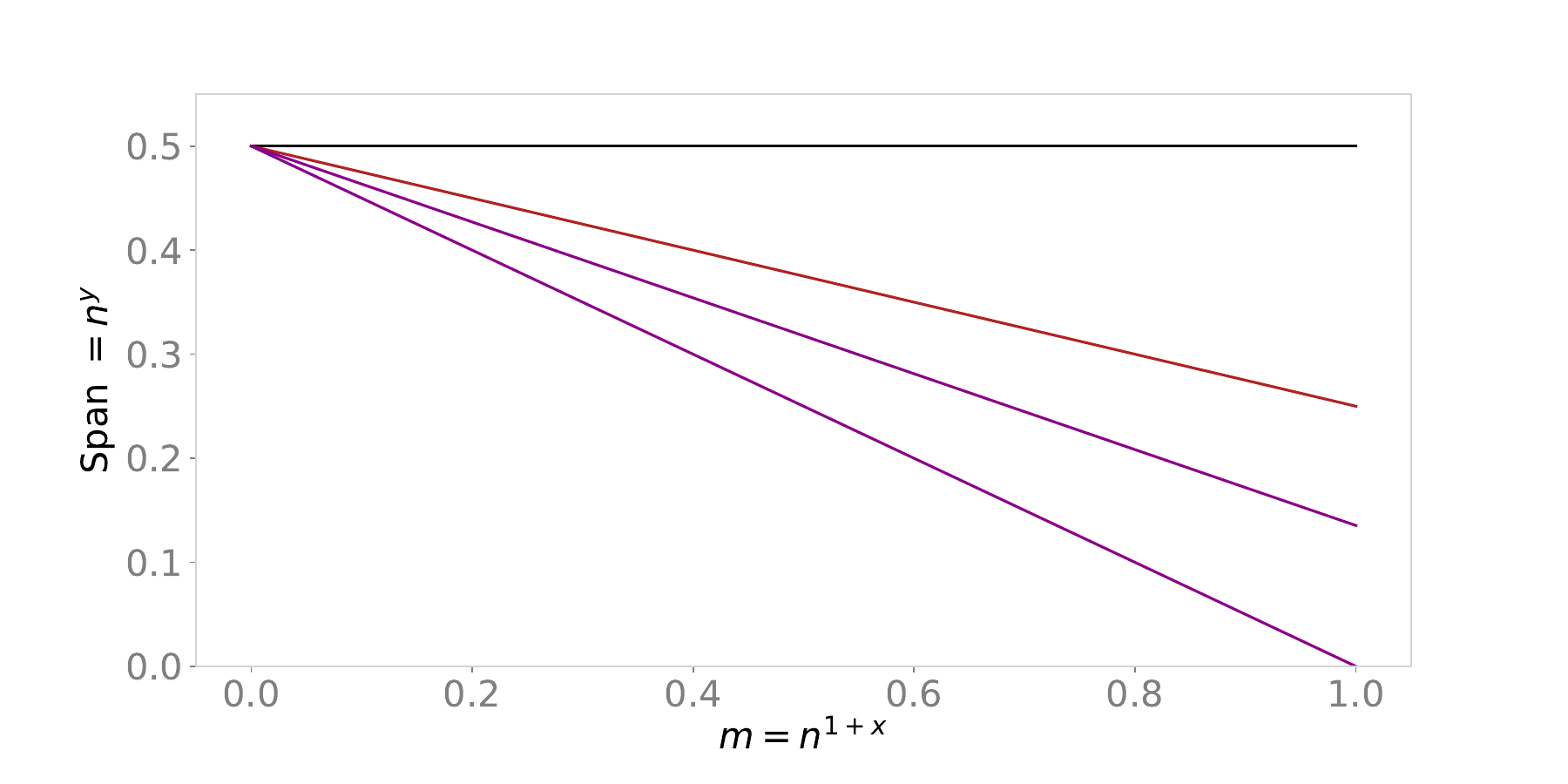}
        \put(315,142){\small{\JLS Parallel Reachability \cite{jambulapati2019parallel} and}}
        \put(315,132){\small{\CFR Parallel SSSP \cite{cao2020efficient,rozhovn2023parallel}}}
        \put(315,80){\textcolor{\cpruneii}{\small{Parallel SSSP (\Cref{thm:par-sssp})}}$\ ^{\dagger}$}
        \put(315,55){\textcolor{\cprunei}{\small{Parallel Reachability (\Cref{thm:par-reach}),}}}
        \put(315,45){\textcolor{\cprunei}{\small{assuming $\omeganc \approx 2.371339$}}}
        \put(315,25){\textcolor{\cprunei}{\small{assuming $\omeganc = 2$}}}
    \end{overpic}
    \caption{
        Comparison of spans of $\Ot(m)$ work parallel algorithms for reachability and SSSP on digraphs.
        Each plotted curve corresponds to such a parallel algorithm.
        The x-axis is the density of the input digraph, and the y-axis is the span of the algorithm.
        $\dagger$: The bound for parallel SSSP also holds for \emph{combinatorial} parallel reachability (i.e. $\omega = 3$).
    }
    \label{fig:plot}
\end{figure}

\paragraph{Parallel Reachability and Shortcut Sets. }
Our main result is on parallel reachability.

\begin{restatable}[Parallel Reachability]{theorem}{parreach}
\label{thm:par-reach}
    Let $\omeganc$ be the fast matrix multiplication exponent in the work for parallel algorithms using polylogarithmic (in $n$) span.
    There is an $\Ot(m)$ work
    and $\sqrt[2\omeganc-2]{n^{\omeganc+o(1)}/m}$ span
    randomized parallel algorithm that, given an unweighted digraph $G$ and source $s \in V$, outputs with high probability all vertices that $s$ can reach and all vertices that reach $s$.
\end{restatable}

Using the fact that $\omeganc < 2.371339$ (by \cite{alman2025more} and \Cref{prop:par-mm}), we can compute reachability in $\Ot(m)$ work and $n^{0.86461}/m^{0.36461}$ span ($n^{0.136}$ span when $m = \Omega(n^2)$).
The previous best span for nearly work-efficient algorithms for reachability was at least $\Omega(\sqrt{n})$, regardless of digraph density: \cite{jambulapati2019parallel} achieves depth $n^{1/2+o(1)}$, while the min-cost flow algorithm of \cite{brand2025parallel} reduces the depth to $\Ot(\sqrt{n})$ when $m \gtrsim n^{1.5}$.

The technical developments used to get the parallel reachability algorithm of \Cref{thm:par-reach} is largely a near-linear time \emph{sequential} construction for a $\sqrt[2\omega-2]{n^{\omega+o(1)}/m}$-shortcut set with size $\Ot(m)$.

\begin{restatable}[Sequential Shortcut Set]{theorem}{seqshort}
\label{thm:seq-shortcut}
    Let $\omega$ be the fast matrix multiplication exponent.
    There is an $\Ot(m)$ time randomized algorithm that, given an unweighted digraph $G$, outputs with high probability a
    $\sqrt[2\omega-2]{n^{\omega+o(1)}/m}$
    -shortcut set $H$ with size 
    $\Ot\paren{m}$.

    More generally, we show the following tradeoff which the above is a special case of.
    Let $\rho \in [\sqrt{n}]$.
    There is an 
    $\Ot\paren{m + n\rho^{2\omega-2}}$
    time randomized algorithm that, given an unweighted digraph $G$, outputs with high probability an $n^{1/2 + o(1)}/\rho$-shortcut set $H$ with size 
    $\Ot(n\rho^2)$.
\end{restatable}

Refer to \Cref{table:omegas} for the spans of our algorithms given for different values of $\omega$ shown in a tabular format, summarizing \Cref{thm:par-reach} and \Cref{thm:seq-shortcut}.
Of particular interest are the extremes where \Cref{thm:seq-shortcut} outputs, if $\omega = 3$, a combinatorially constructed $n^{3/4 + o(1)}/m^{1/4}$-shortcut set and, if $\omega = 2$ an $n^{1+o(1)}/\sqrt{m}$-shortcut set almost matching the tunable folklore shortcut set in its parameters.
Also of note: the tradeoff given by \Cref{thm:seq-shortcut} is strictly better than that given by \cite{cao2021brief}, which is a randomized $\Ot\paren{m\rho^2 + n\rho^4}$ time construction of an $n^{1/2 + o(1)}/\rho$-shortcut set $H$ with size $\Ot(n\rho^2)$; most crucially, because of the $m\rho^2$ term, their tradeoff could not achieve near-linear work for any graph density.

Parallelizing the previous state-of-the-art near-linear size $n^{1/2+o(1)}$-shortcut set of \cite{jambulapati2019parallel} was, while important, somewhat tedious.
Using the new technology of \cite{blackbox}, we can easily parallelize \Cref{thm:seq-shortcut} --- see \Cref{thm:par-shortcut} for details --- and consequently use the parallel construction to immediately get our main result \Cref{thm:par-reach}.

\begin{table}
    \centering{
        \begin{tabular}{|c|c|c|}
            \hline 
            \textbf{Value of $\omega$} & Span (or $\beta$) & Span (or $\beta$) when $m=\Theta(n^{2})$\tabularnewline
            \hline 
            \hline 
            \textbf{(Combinatorial): $3$} & $n^{3/4 + o(1)}/m^{1/4}$ & $n^{1/4 + o(1)}$\tabularnewline
            \hline 
            \textbf{(Ideal): $2$} & $n^{1+o(1)}/\sqrt{m}$ & $n^{o(1)}$\tabularnewline
            \hline 
            \textbf{(Current \cite{alman2025more}): $2.371339$} & $n^{0.86461}/m^{0.36461}$ & $n^{0.136}$\tabularnewline
            \hline 
        \end{tabular}
    }
    \caption{Span of our near-linear work parallel reachability algorithms (or $\beta$ value of our $\beta$-shortcut sets) for different values of $\omega$.}
    \label{table:omegas}
\end{table}

\paragraph{Parallel SSSP and Hopsets. }
We then show analogous results for SSSP, albeit with slightly weaker bounds.
\begin{restatable}[Parallel SSSP]{theorem}{parsssp}
\label{thm:par-sssp}
    There is an $\Ot(m)$ work
    and $\sqrt[4]{n^{3 + o(1)}/m}$ span
    randomized parallel algorithm that, given a polynomially bounded non-negative integer weighted digraph $G$ and source $s \in V$,
    outputs with high probability
    $\dist(s,v)$ for all $v \in V$ along with a shortest path tree rooted at $s$.
\end{restatable}

On dense digraphs, this gives a nearly work-efficient $n^{1/4+o(1)}$ span algorithm for SSSP.
The previous best span for nearly work-efficient algorithms for SSSP was at least $\Omega(\sqrt{n})$, similar to the previous state of reachability: \cite{rozhovn2023parallel,cao2023parallel} achieve depth $n^{1/2+o(1)}$, while the min-cost flow algorithm of \cite{brand2025parallel} reduces the depth to $\Ot(\sqrt{n})$ when $m \gtrsim n^{1.5}$.

Analogously to parallel reachability, the main driver of \Cref{thm:par-sssp} is a near-linear time \emph{sequential} construction of a $\paren{\sqrt[4]{n^{3 + o(1)}/m}, \eps}$-hopset $H$ with size $\Ot(m/\eps^2)$.

\begin{restatable}[Sequential Hopset]{theorem}{seqhop}
\label{thm:seq-hopset}
    There is an $\Ot(m/\eps^2)$ time randomized algorithm that, given a 
    polynomially bounded non-negative integer 
    weighted digraph $G$, outputs with high probability a $\paren{\sqrt[4]{n^{3 + o(1)}/m}, \eps}$-hopset $H$ with size $\Ot(m/\eps^2)$.

    More generally, we show the following tradeoff which the above is a special case of.
    Let $\rho \in [\sqrt{n}]$.
    There is an $\Ot\paren{m/\eps^2 + n\rho^4}$ time randomized algorithm that, given a  
    polynomially bounded \ifthenelse{\boolean{restate}}{}{\footnote{Since we ultimately intend to compare our SSSP bounds with \cite{rozhovn2023parallel}, we assume polynomially bounded edge weights the same way they have. Without this polynomial bound, our running time and hopset size suffers a $\log(nW)$ factor where $W$ is the largest weight in $G$. This holds, too, for \Cref{thm:par-hopset}.}}
    non-negative integer weighted digraph $G$, outputs with high probability a $\paren{n^{1/2 + o(1)}/\rho, \eps}$-hopset $H$ with size $\Ot\paren{n/\eps^2 + n\rho^2}$.
\end{restatable}

In particular, \Cref{thm:seq-hopset} on dense digraphs runs in $\Ot(n^2/\eps^2)$ time and shortcuts $(1+\eps)$-approximate shortest paths to $n^{1/4 + o(1)}$ hops, whereas the previous best hop bound was $n^{1/2 + o(1)}$ given by \cite{cao2020efficient} for all density regimes.
Similarly to shortcut sets, the tradeoff of \Cref{thm:seq-hopset} is strictly better than (the sequential version) of \cite{cao2021brief}.

We get \Cref{thm:par-sssp} by parallelizing \Cref{thm:seq-hopset}, again simply, using the new technology of \cite{blackbox} --- see \Cref{thm:par-hopset} for details on said parallelization.

Due to technical complications described later in \Cref{rmk:hopset}, not all of the ideas used for shortcut sets (\Cref{thm:seq-shortcut}) go through for constructing hopsets (\Cref{thm:seq-hopset}).
Nevertheless, the main conceptual ideas do indeed go through, allowing us to match the bound in the combinatorial version of \Cref{thm:seq-shortcut}, i.e. when $\omega = 3$ (see \Cref{table:omegas} for more on the combinatorial version).

\paragraph{Related Work on Shortcut Sets and Hopsets. }
There exists a rich literature on shortcut sets and hopsets with stronger existential bounds than the folklore construction \cite{kogan2022new,kogan2022beating,bernstein2023closing,kogan2023faster,kogan2023towards} though currently no near-linear time construction of is known, even in the sequential setting.
In particular, the breakthrough of \cite{kogan2022new} showed the existence of a $\sqrt[3]{n^2/m}$-shortcut set with size $\Ot(m)$ which, if constructible in near-linear work, would give a depth $\sqrt[3]{n^2/m} \ll n/\sqrt{m}$ algorithm for reachability.
Along another direction, \cite{bals2025greedy} recently showed a deterministic construction of an $O(\sqrt{n})$-shortcut set with size $O(n \log^*n)$ in almost-linear time.
We include a more detailed discussion of these related works in \Cref{sec:open-problems} (open problems).

\subsection{Organization}

Preliminaries are covered in \Cref{sec:preliminaries}.
A high-level overview of our main technical ideas is then provided in \Cref{sec:overview}.
The core of the parallel reachability result is in \Cref{sec:seq-shortcut} (which goes over a sequential construction of shortcut sets) and finished up in \Cref{sec:reach} where parallelization of the sequential construction is discussed.
Similarly, the core of the parallel SSSP result is in \Cref{sec:seq-hopset} (which goes over a sequential construction of hopsets) and finished up in \Cref{sec:sssp}.
Finally, we leave open problems to \Cref{sec:open-problems}.

\setboolean{restate}{true}

%% file: preliminaries.tex
\section{Preliminaries}
\label{sec:preliminaries}

\paragraph{Numbers and Sets. }
For any positive integer $z \in \IN$, we use $[z]$ to denote $\set{1,2,\ldots,z}$.
For integers $a \le b$, we use $[a,b]$ to denote $\set{a, a+1, \ldots, b}$.
We also use the notation $[a,b] + c$ to denote $[a+c, b+c]$ and $k[a,b]$ for $[ka,kb]$.

\paragraph{Graphs. }
In this paper we work with unweighted digraphs $G=(V,E)$ and also weighted digraphs $G=(V,E,w)$.
We use $V(G)$ (or when unambiguous, $V$) to refer to the vertex set of $G$ and $E(G)$ (or when unambiguous, $E$) to refer to the edge set of $G$.
An edge pointing from vertex $u$ to vertex $v$ is denoted with $uv$.
For any $V' \subseteq V$ we use $G[V']$ to denote the subgraph of $G$ induced on $V'$.
When it is clear, we use $n = \card{V}$ and $m = \card{E}$.

\paragraph{Paths and Distances. }
A path using $h$ hops is a sequence of vertices $P = \langle v_0, v_1, \ldots, v_h \rangle$ where $v_{i-1}v_i \in E$.
A subpath of $P = \langle v_0, v_1, \ldots, v_j \rangle$ is a path of the form $\langle v_a, v_{a+1}, \ldots, v_z \rangle$ with $0 \le a \le z \le h$.
If we say a path $P$ is split into $k$ subpaths $P_1, P_2, \ldots, P_k$ we mean that the $P_i$ are disjoint and $P = P_1 P_2 \ldots P_k$.
In a weighted digraph, the length of the path is $w(P) = \sum_{i} w(v_{i-1} v_i)$.
The distance $\dist(s,t)$ is the length of the shortest path from $s$ to $t$, and the hop bounded distance $\bdist{h}(s,t)$ is the length of the shortest path from $s$ to $t$ which uses at most $h$ hops; we sometimes subscript these operators with the graph we are taking distances over (e.g. $\dist_G(s,t)$).

\paragraph{Hopbounds. }
We say that a digraph $G$ has \emph{reachability hopbound} $h$ if for all $s,t \in V$ such that $s$ reaches $t$, there is a path using at most $h$ hops from $s$ to $t$.
We say that $G$ has $(1+\eps)$-approximate shortest path hopbound $h$ if, for all $s,t \in V$, we have $\bdist{h}(s,t) \le (1+\eps)\dist(s,t)$.

\paragraph{Shortcut Sets and Hopsets. }
$H \subseteq V^2$ is a \emph{$\beta$-shortcut set} of a digraph $G$ if {\bf 1)} for all $st \in H$, which we call shortcuts, there is an $st$-path in $G$ and {\bf 2)}
The reachability hopbound in $G \cup H$ is not more than $\beta$.
$H \subseteq V^2$ is a \emph{$(\beta, \eps)$-hopset} of a digraph $G$ if for all $s, t \in V$ the following is satisfied:
\begin{align*}
    \dist_G(s,t) \le \bdist{\beta}_{G \cup H}(s,t) \le (1+\eps)\dist_G(s,t).
\end{align*}

\paragraph{Reachability Relations and Relevant Vertices.}
If there is a path from $s$ to $t$ we say $s$ reaches $t$ and also $t$ is reached by $s$; we say here that $s$ and $t$ are related.
$\desc{G}{s}$ is the set of all vertices $s$ reaches in $G$ and $\anc{G}{t}$ is the set of all vertices which reach $t$ in $G$.
We denote the set of \emph{relevant vertices} to $v$ with $\relevant{G}{v} = \desc{G}{v} \cup \anc{G}{v}$.
These are extended to sets in the natural way; for example, $\desc{G}{S} = \cup_{s \in S} \desc{G}{s}$.

We also extend the above notions to bounded distance analogues.
If $\dist(s,t) \le d$ we say $s \preceq_d t$ and also that $s$ and $t$ are $d$-related, $s$ $d$-reaches $t$, and $t$ is $d$-reached by $s$.
Then, $\ddesc{G}{s}{d} = \set{t : s \preceq_d t}$ and $\danc{G}{t}{d} = \set{s : s \preceq_d t}$ and $\drelevant{G}{v}{d} = \ddesc{G}{v}{d} \cup \danc{G}{v}{d}$.
We similarly extend these to sets.

\paragraph{Path-related Vertices (Ancestors, Descendants, Bridges). }
For any path $P$, we call $\anc{G}{P} \setminus \desc{G}{P}$ its \emph{ancestors} and $\desc{G}{P} \setminus \anc{G}{P}$ its \emph{descendants} and finally $\anc{G}{P} \cap \desc{G}{P}$ its \emph{bridges}.
Similarly, we call $\danc{G}{P}{d} \setminus \ddesc{G}{P}{d}$ its \emph{$d$-ancestors} and $\ddesc{G}{P}{d} \setminus \danc{G}{P}{d}$ its \emph{$d$-descendants} and finally $\danc{G}{P}{d} \cap \ddesc{G}{P}{d}$ its \emph{$d$-bridges}.

\paragraph{Fast Matrix Multiplication. }
We use the following result for parallel matrix multiplication, which states that we can multiply matrices with work matching the current best sequential time complexity of \cite{alman2025more} and span only $O(\log n)$.
\begin{proposition}[Paraphrasing (with some modification) of Theorem 5.7 Part 1 in \cite{jaja1992parallel}]
\label{prop:par-mm}
    Let $A$ and $B$ be two $n$ by $n$ matrices with entries in a ring $R$ with operations $+,\times$.
    The matrix $AB$ can be computed by a parallel algorithm with $O(\log n)$ span and $O(M(n))$ work, where $M(n)$ is the best known sequential bound for computing $AB$ over $R$ by an algorithm that can be written as an algebraic circuit with $+,\times$ gates.
\end{proposition}
That is to say, throughout this paper $\omeganc < 2.371339$ (even in the parallel setting) which is established by \cite{alman2025more}, whose (sequential) algorithm satisfies the premise of \Cref{prop:par-mm}.

\paragraph{Probability. }
We use the following one-sided Chernoff bound for sums of independent $\set{0,1}$ random variables $X_i$ with mean $\mu = \expect{\sum_i X_i}$:
\begin{align*}
    \prob{X \ge \paren{1 + \delta}\mu} \le e^{-\delta^2\mu/(2 + \delta)} \textrm{ for } \delta \ge 0.
\end{align*}

%% file: overview.tex
\section{High-Level Overview}
\label{sec:overview}

Here we give a sketch of the main ideas in this paper.
To obtain our result, we introduce a new pruning step into the framework of \cite{jambulapati2019parallel,cao2020efficient}, respectively called \prunei and \pruneii.
The main technical contribution is actually in the analysis of this.
We sketch a new top-down analysis of \cite{jambulapati2019parallel,cao2020efficient} which naturally suggests how to apply the pruning steps, leading to the improvements below.
Since the core insights can be found within our reachability result, which is much simpler, we omit any discussion of SSSP in this overview.

\subsection{Summary of Prior Work}
By using standard techniques, the reachability problem on a digraph $G$ is reduced to constructing a shortcut set of $G$; indeed, if one can construct a $\beta$-shortcut set $H$ in $\Ot(m)$ work and $D$ span, then one gets a parallel reachability algorithm by running a parallel BFS from the source $s$ on $G \cup H$, which takes an additional $\Ot(m + \card{H})$ work and $\beta$ span (so, in sum, an $\Ot(m + \card{H})$-work and $O(D + \beta)$-span algorithm).
We thus focus on the problem of constructing a shortcut set.

\cite{jambulapati2019parallel} refines the breakthrough of \cite{fineman2018nearly} to give an $\Ot(m)$ work algorithm which constructs an $n^{1/2 + o(1)}$-shortcut set in span $n^{1/2 + o(1)}$.
As a precursor to this, they give a sequential $\Ot(m)$ time algorithm which constructs an $n^{1/2 + o(1)}$-shortcut set with $\Ot(n)$ size, henceforth called the \JLS shortcut set.
We first turn our attention to this sequential construction.

Let us assume $G$ is a directed acyclic graph (DAG).
Suppressing some details (e.g. parameter settings) that are inessential to an overview, the \JLS algorithm can be described recursively in the following way, where $G$ is the level $0$ recursive instance.
On a level $r$ recursive instance $G[V']$, where $V' \subseteq V$, some pivot vertices $S \subseteq V'$ are selected randomly.
For each $v \in V'$ and $p \in S$, if $v$ reaches $p$ then add $vp$ to the shortcut set, and if $p$ reaches $v$ then add $pv$ to the shortcut set.
Next, we partition $V' = V'_1 \sqcup V'_2 \sqcup \ldots$ into an equivalence relation based on the reachability relation to $S$ in $G[V']$; that is, for all $u,v \in V'$ and $i$, we have $u,v \in V'_i$ iff the set of pivots in $S$ that reach $u$ and $v$ are equal and the set of pivots in $S$ that $u$ and $v$ reach are equal.
We then set each $G[V'_i]$ to be a level $r+1$ recursive instance.
A complete description of \JLS is given later in \Cref{sec:seq-shortcut}.

For appropriately chosen parameters, the above algorithm runs in $\Ot(m)$ time and yields an $n^{1/2 + o(1)}$-shortcut set with near-linear size.

\paragraph*{Framework for Bounding the Diameter. }
To bound the diameter of \JLS, we fix an arbitrary path $P$ and show that $P$ is shortcut to length $n^{1/2 + o(1)}$.
First, observe that at any fixed recursion level, say $r$, the path $P$ is split into contiguous subpaths $P_1, P_2, \ldots$ each belonging to distinct level $r+1$ recursive instances $G_1, G_2, \ldots$ respectively.
To see this, note that if a pivot vertex reaches (resp.~is reached by) any two vertices in $P$, then it reaches (resp.~is reached by) every vertex on the subpath between said two vertices in $P$; thus, if $x$ and $y$ are in the same subproblem $G_i$, then so are all the vertices on any path from $x$ to $y$. (See \Cref{obs:s-splitting} for the formal proof.)

We use this observation to define the \emph{subproblem tree of $P$}, which tracks how $P$ is shortcut.
(i) The root node\footnote{Here and throughout, we use ``node'' to refer to the nodes of subproblem trees, which are analysis tools which our algorithms are not aware of, and ``vertex'' to refer to the vertices of digraphs, which are the objects our algorithms interface with.} of this tree, at level $0$, is $P$; (ii) If a bridge of a level $r$ node $P'$ is selected in the recursion level $r$ instance $G'$ (in which $P'$ is contained) then $P'$ is a leaf node and, otherwise, $P'$ has level $r+1$ children nodes $P'_1, P'_2, \ldots$ corresponding to how $P'$ is split.
The reader should observe that 
each leaf $P'$ is shortcut to $O(1)$ hops (moving through the bridge of $P'$) and, hence, the number of nodes in the subproblem tree upper bounds the number of hops $P$ is shortcut to (up to a constant factor).

\cite{jambulapati2019parallel} shows that the number of nodes in the subproblem tree of any path $P$ is bounded above by $n^{1/2+o(1)}$.
Crucially, they use the following key lemma within a bottom-up argument by induction:
\begin{align*}
    \expect{\sum_{i = 1}^{\card{S} + 1} \card{\relevant{G'_i}{P'_i}} \;\middle|\; \card{S}} \le \frac{2}{\card{S} + 1} \card{\relevant{G'}{P'}},
    \tag{\Cref{prop:jls-2}}
\end{align*}
where $P'$ (in subproblem $G'$) is an arbitrary internal node in the subproblem tree and $S$ are the pivots selected in $\relevant{G'}{P'}$ that split $P'$ into subpaths $P'_1, P'_2, \ldots, P'_{\card{S}+1}$ in subproblems $G'_1, G'_2, \ldots, G'_{\card{S}+1}$ respectively.
Recall (from \Cref{sec:preliminaries}) that $\relevant{G'}{P'}$ are the vertices in $G'$ that can reach or are reached by some vertex in $P'$.

We suggest that \Cref{prop:jls-2} is best interpreted in the following way.
Intuitively, splitting a path $P'$ into $P'_1, P'_2, \ldots, P'_{\card{S}+1}$ is detrimental to shortcutting $P'$ (and hence $P$) since the algorithm will no longer be able to add any shortcut between $P'_i$ and $P'_j$ for $i \neq j$.
For example, if $\card{S} \gg n^{1/2+o(1)}$, then there would no longer be any hope of shortcutting $P$ to $n^{1/2+o(1)}$ hops.
On the other hand, \Cref{prop:jls-2} shows that splitting a path still offers progress: the number of relevant vertices to $P'$ drops significantly.
Since $P' \subseteq \bigcup_{i = 1}^{\card{S} + 1} \relevant{G'_i}{P'_i}$ and each $v \in P'$ is a bridge of some $P'_i$, the fraction of bridges increases\footnote{This is not entirely true, since bridges of $P'$ may cease to be bridges of $P'_i$ for any $i$.}.
It is consequently more likely to select bridges as pivots and shortcut the subpaths that $P'$ is split into to $O(1)$ hops each.

\subsection{Summary of Our Main Result}

We now sketch a new proof of the diameter bound of \JLS, which reveals an opportunity to lessen the diameter of the shortcut set to $\beta \ll n^{1/2 + o(1)}$ at the expense of having more edges in the shortcut set.
Thereafter, we will show concretely how to exploit this opportunity by adding just one line to \JLS.

\paragraph*{A New, Top-Down, More General Analysis of the \JLS Diameter Bound. }
For simplicity's sake, let us assume something stronger than \Cref{prop:jls-2}: that the typical scenario occurs \emph{deterministically}.
\begin{align*}
    \sum_{i = 1}^{\card{S} + 1} \card{\relevant{G'_i}{P'_i}} \le \frac{2}{\card{S} + 1} \card{\relevant{G'}{P'}}.
    \tag{Idealized \Cref{prop:jls-2}}
\end{align*}
Assuming Idealized \Cref{prop:jls-2} above does not change any structure of our argument and only removes clutter related to formalization of probabilistic guarantees.

Given this, we first bound the number of nodes in the subproblem trees.
\begin{claim}
\label{claim:1}
    For all $x$, there are at most $n^{1/2 + o(1)}/\sqrt{x}$ nodes $P'$ in the subproblem tree where $\card{\relevant{G'}{P'}} \ge x$.
\end{claim}
\begin{proof}[Proof Sketch]
Define the function $\phi(P') = \sqrt{|\relevant{G'}{P'}|}$.
We will show $\sum_{\text{tree nodes } P'} \phi(P') = n^{1/2 + o(1)}$ and from there the bound follows since each $P'$ with $\card{\relevant{G'}{P'}} \ge x$ contributes $\phi(P') \ge \sqrt{x}$ to the sum.

To see that $\sum_{P'} \phi(P') = n^{1/2 + o(1)}$, observe that $\phi(P) = O(\sqrt{n})$; the subproblem tree has $O(\log n / \log\log n)$ levels (we have not justified this in the overview, see \Cref{sec:seq-shortcut}); and
\begin{align*}
    \sum_{i}^{\card{S}+1} \phi(P'_i) 
    =
    \sum_{i}^{\card{S}+1} \sqrt{\card{\relevant{G'_i}{P'_i}}} 
    \le
    \sqrt{(\card{S}+1)\sum_{i}^{\card{S}+1}\card{\relevant{G'_i}{P'_i}}}
    \le
    \sqrt{2} \phi(P')
\end{align*}
where the first inequality is by Cauchy-Schwarz and the second is by Idealized \Cref{prop:jls-2}.
\end{proof}

By plugging in $x = 1$, we recover the $n^{1/2 + o(1)}$ diameter bound of \JLS.
Even more, we may plug in other values of $x$ to get more refined upper bounds.
Noting that the fanout of the subproblem tree is polylogarithmic (another detail omitted here), this suggests the following type of strategy: shortcut all nodes $P'$ with $\card{\relevant{G'}{P'}} < x$ to $O(1)$ hops.
We can then consider the \emph{pruned subproblem tree of $P$} with all nodes $P'$ where $\card{\relevant{G'}{P'}} < x$ removed.
By \Cref{claim:1}, the number of nodes in the pruned tree is at most $n^{1/2 + o(1)}/\sqrt{x}$ up to polylogarithmic factors, which upper bounds the number of hops to which $P$ is shortcut.

\paragraph*{A Simple Combinatorial Improvement (Warmup). }
Assume the input digraph is dense -- that is, it has $\Omega(n^2)$ edges.
Denote $\textrm{BFS}_{\sqrt{n}}(G',v)$ as the first $\sqrt{n}$ vertices found in a BFS in $G'$ starting from $v \in V(G')$.

Here is a simple modification to the \JLS algorithm.
For each recursive subproblem $G'$ and each $v \in V(G')$, add all the following edges to the shortcut set:
edges from $v$ to $\textrm{BFS}_{\sqrt{n}}(G',v)$.

Note that $\textrm{BFS}_{\sqrt{n}}(G',v)$ runs in $O(n)$ time, totaling up to $\Ot(n^2)$ time over all the calls made in all recursive subproblems.
This additional step thus gives a negligible overhead to the running time of \JLS.

Crucially, for any node $P' = \langle v_0,v_1,\ldots,v_h\rangle$ in the subproblem tree with $\card{\relevant{G'}{P'}} < \sqrt{n} = x$, the edge $v_0 v_h$ will be added to the shortcut set since $v_h \in \textrm{BFS}_{\sqrt{n}}(G',v_0)$.
Namely, $P'$ is shortcut to $O(1)$ hops and is removed from the pruned subproblem tree which, accordingly, has $O(n^{1/4 + o(1)})$ nodes by \Cref{claim:1} (which, recall, upper bounds the length $P$ is shortcut to).

In summary, this simple modification to \JLS gives a near-linear time construction of an $O(n^{1/4 + o(1)})$-shortcut set on dense digraphs $G$.

\paragraph*{A Stronger but Non-Combinatorial Improvement. }
Our main result for reachability does not use the combinatorial improvement above, but instead utilizes fast matrix multiplication.
Here we just briefly summarize the idea and leave the details to \Cref{sec:seq-shortcut}.
Instead of computing $\textrm{BFS}_{\sqrt{n}}(G',v)$ from each $v \in V(G')$ in every recursive subproblem $G'$ as in our warmup, we get even more of a speed up from loosely speaking computing shortcuts in an ``all-pairs'' fashion (i.e. transitive closures).
This is easier said than done since one call to a transitive closure algorithm is more expensive than one BFS call; we must thus be judicious of when and on what vertices to call the transitive closures (instead of from every vertex like we did with BFS).

When $\card{\relevant{G'}{P'}}$ is small (at the threshold where we wish to prune away the node $P'$ from the subproblem tree of $P$), one might try to call a transitive closure on a ball centered at $v$ for some $v \in P'$.
The flaw in this idea is that the choices made by the algorithm must be oblivious to $P'$.
A more oblivious approach would be to call a transitive closure on $\relevant{G'}{p}$ for some pivot $p \in \relevant{G'}{P'}$.
However, $\card{\relevant{G'}{p}}$ is not necessarily the same as $\card{\relevant{G'}{P'}}$ since $p$ may not be a vertex in $P'$.
In fact, $p$ may reach or be reached by many vertices that have no reachability relation to $P'$, and hence $\card{\relevant{G'}{p}}$ might be very large despite $\card{\relevant{G'}{P'}}$ being very small; calling a transitive closure on $\card{\relevant{G'}{p}}$ may then be too expensive!

Instead, we do the following.
For every pivot $p$, if $\card{\relevant{G'}{p}}$ is small enough, we add all edges in the transitive closure of $\relevant{G'}{p}$ to the shortcut set.
This can be done in $\Ot(\card{\relevant{G'}{p}}^{\omega})$ time by repeatedly squaring the adjacency matrix of $G'[\relevant{G'}{p}]$.
By virtue of $\card{\relevant{G'}{p}}$ being small, we have control over the running time of this improvement.
Moreover, since all edges in the transitive closure of $\relevant{G'}{p}$ are added to the shortcut set, $p$ would not contribute any children to the node $P'$.
But would it allow us to prune the node $P'$ completely when $\card{\relevant{G'}{P'}}$ is small?
Not quite.
We show in \Cref{sec:seq-shortcut} that it allows us to prune nodes in the subtree rooted at $P'$ so that the pruned subtree has $O(\log n)$ nodes;
loosely, the argument is as follows: for any node $P''$ in this subtree, if all $p \in \relevant{G''}{P''}$ have small $\card{\relevant{G''}{p}}$ then $P''$ has at most one child and, otherwise, there is some $p \in \relevant{G''}{P''}$ with large $\card{\relevant{G''}{p}}$, certifying that it is very unlikely for any vertex in $\relevant{G''}{P''}$ to be sampled as a pivot.
This yields the $n^{0.136}$-shortcut set in near-linear time for dense graphs.

\paragraph*{Parallelization. }
Finally, we parallelize the shortcut set construction by using a blackbox framework provided by \cite{blackbox}.
Morally, the framework says that (up to some fudging) we need only guarantee that a parallel algorithm for computing a $\beta$-shortcut set runs in $\Ot(\beta)$ span on DAGs with reachability hopbound $\Ot(\beta)$;
this is enough to show that a $\beta$-shortcut set runs in $\Ot(\beta)$ span on any digraph,
regardless of its diameter.
It is easy to show that our sequential construction for $\beta$-shortcut sets run in $\beta$ span on DAGs with reachability hopbound $\Ot(\beta)$, since BFS calls will run in span $\beta$.

%% file: shortcut-set.tex
\section{Sequential Shortcut Set Construction}
\label{sec:seq-shortcut}

In this section we prove \Cref{thm:seq-shortcut} which pertains to a sequential construction of shortcut sets.
\Cref{subsec:jls} gives an overview of the \JLS shortcut set and outlines a strategy for how to improve upon it.
We then go over the details in \Cref{subsec:prune-1} and finish things up in \Cref{subsec:shortcut-analysis}.

\input{jls}
\input{prune-1}
\input{shortcut-analysis}

%% file: jls.tex
\subsection{The JLS Shortcut Set, and a Strategy for Improved Bounds}
\label{subsec:jls}

The \JLS shortcut set, which is a refinement of the construction of Fineman~\cite{fineman2018nearly}, can be summarized at a high level as follows.
We first assume the input $G$ is a DAG (since it is easy to compute strongly connected components (SCCs) in a digraph in linear time and, adding $O(n)$ edges to the shortcut set, shortcut each SCC to two hops).
There is a global parameter $k$, which we can think of as $\log n$, that controls a sampling rate and recursion depth $\log_k(n)$.
The construction is found via a recursive algorithm where at recursion level $r$, around $k^{r+1}$ pivot vertices from the base graph $G$ are selected uniformly at random (these pivots are divided up possibly unequally among all level $r$ subproblems on graphs $G[V_1], G[V_2], \ldots$).
Let us focus on one level $r$ subproblem, say, $G[V_1]$.
Reachability within $G[V_1]$ is computed for the pivots belonging to $V_1$ and based on this the following actions are made:
\begin{itemize}
    \item For each pivot $p$, add the edges $vp$ (resp. $pv$) to the shortcut set if $v$ reaches (resp. is reached by) $p$.
    \item For each $v \in V_1$, give it a set of labels based on its reachability relation to the pivots.
        That is, for all pivots $p \in V_1$ and $v \in V_1$ (i) if $p$ reaches $v$, give $v$ the label ``$p$ reaches me''; (ii) if $v$ reaches $p$, give $v$ the label ``I reach $p$''; (iii) otherwise give $v$ the label ``I have no relation to $p$''.
    \item Recurse into $V_{1,1}, V_{1,2}, \ldots \subseteq V_1$ where the $V_{1,i}$'s are an equivalence class in $V_1$ using the labeling above.
        That is, $x, y \in V_{1,i}$ iff $x$ and $y$ have the exact same labels.
\end{itemize}
The aggregate of added shortcuts forms the \JLS shortcut set.
Below we give a formal description of the algorithm.

\begin{tbox}
    \algname{\JLS}\\
    \textbf{Global Parameters:} $k$ is a global parameter to be fixed later. $n$ is the number of vertices in the base input graph (it thus remains fixed through all recursive calls). \\
    \textbf{Input:} A DAG $G = (V, E)$ and a recursion level $r$.\\
    \textbf{Output:} A shortcut set $H \subseteq V^2$.
    \begin{enumerate}
        \item Randomly and independently sample, with probability $p_r \gets \frac{100k^{r+1}\log n}{n}$, each $v \in V$ as a pivot.
            Let $S$ be the set of pivots.
        \item For each $p \in S$, compute $\overset{\textcolor{gray}{\textrm{reach }p}}{\anc{G}{p}}$, $\overset{\textcolor{gray}{p \textrm{ reaches}}}{\desc{G}{p}}$
        and:
            \begin{itemize}
                \item Add $vp$ to $H$ for all $v \in \anc{G}{p}$.
                \item Add $pv$ to $H$ for all $v \in \desc{G}{p}$.
                \item Add $\vanc{p}$ label to all $v \in \anc{G}{p}$.
                \item Add $\vdesc{p}$ label to all $v \in \desc{G}{p}$.
            \end{itemize}
        \item $V_1, V_2, \ldots, V_t \gets$ Partition of all $v \in V$ such that $x, y \in V_i$ iff $x$ and $y$ have the exact same labels.
        \item Output $H \cup \paren{\bigcup_{i \in [t]} \JLS \textrm{ on } G[V_i], r+1}$.
    \end{enumerate}
\end{tbox}

\begin{proposition}[Paraphrasing Theorem 5 from \cite{jambulapati2019parallel}]
\label{prop:jls-main}
    \JLS runs in time $\Ot(mk)$ and produces an $n^{1/2 + O(1/\log k)}$-shortcut set of size $\Ot(nk)$ with probability at least $1 - O(n^{-10})$.
\end{proposition}

The above is proved in \cite{jambulapati2019parallel}, with the running time and size bounds following quite easily from \Cref{prop:jls-1} (stated later) and Chernoff bounds.
Modulo a few key statements, we will give an alternate (and more extendable) proof of the $n^{1/2 + O(1/\log k)}$ diameter bound in \Cref{subsec:prune-1}.
For now, let us try to better understand the \JLS shortcut set at a high level and, from this, outline a strategy for how to get a better shortcut set.

\paragraph{The Subproblem Tree. }
The way the diameter of the shortcut set is upper bounded comes from \cite{fineman2018nearly}.
There, an arbitrary path $P$ is selected for the sake of analysis.
Notice first that $P$ is split into contiguous subpaths in the recursive calls and the same is true of its subpaths, and so on.
More specifically, observe the following (shown in \cite{jambulapati2019parallel}).
\begin{observation}
\label{obs:s-splitting}
    Let $P'$ be a subpath of $P$, contained in a level $r$ recursive instance $G'$.
    If $z - 1$ of the non-bridge vertices in $\relevant{G'}{P'}$ are chosen as pivots at level $r$, then $P'$ is split into at most $z$ disjoint subpaths $P'_1, \ldots, P'_z$ belonging to distinct level $r+1$ recursive instances $G'_1, \ldots, G'_z$.
\end{observation}
\begin{proof}
    Let $P' = \langle v_0, v_1, \ldots, v_h \rangle$ and $S$ be the set of pivots selected from $\relevant{G'}{P'}$.
    If $p \in S$ reaches $v_i$, then it reaches $v_j$ for $j > i$.
    Similarly, if $p \in S$ is reached by $v_i$, then it is reached by $v_j$ for $j < i$.
    Thus, if $v_i$ and $v_j$ are related to $S$ in the same way, then $v_i, v_{i+1}, \ldots, v_j$ all have the same relation to $S$.
    Since there are at most $z - 1$ locations where the reachability relation can change, there are at most $z$ subpaths.
\end{proof}

We will examine how $P$ evolves (i.e. is split and shortcutted) through the execution of \JLS more carefully.
To do this, we think of the \emph{subproblem tree} of $P$.
Each node of the subproblem tree is associated with some subpath of $P$, and the tree can be described as:
\begin{itemize}
    \item Root node $P$.
        This is the $0$th level of the tree.
    \item For each node $P'$ in level $r$ of the tree (i.e. $P'$ is contained in a level $r$ recursive instance $G'$), if a bridge of $P'$ in $G'$ is sampled at level $r$, then $P'$ is a leaf.
        Note that in this case $P'$ has been shortcut to $2$ hops since, denoting $P'=\langle v_1,v_2,\ldots,v_h\rangle$ and the sampled bridge $b$, the edges $v_0b$ and $bv_h$ are added to the shortcut set.
    \item Otherwise, the path $P'$ is split in the instance $G'$ into subpaths $P'_1, \ldots, P'_z$ belonging to level $r+1$ recursive instances $G'_1, \ldots, G'_z$.
        The node $P'$ will have $z$ children $P'_1, \ldots, P'_z$ at level $r+1$ of the tree.
\end{itemize}
It is important to note that the algorithm is not aware of this subproblem tree and it is merely a tool for our analysis.
Crucially, the number of nodes in this subproblem tree is (up to a constant factor) an upper bound on the number of hops that $P$ is shortcutted to since we can traverse from the start of $P$ to the end along the leaves (which have been shortcutted to $2$ hops) and the edges joining the leaves.

\paragraph{Our Strategy. }
\cite{jambulapati2019parallel} shows that the subproblem tree of any path $P$ has at most $n^{1/2 + O(1/\log k)}$ nodes.
We employ algorithmic tactics to prune away nodes from this tree in our analysis, to the extent that the tree we analyze has a substantially smaller number of nodes.
For example, a node $P'$ in recursive instance $G'$ may have a large subtree in $P$'s subproblem tree.
If we are able to identify some structure in $G'$ so that $P'$ is immediately shortcutted to $O(1)$ hops, we are then permitted to ignore the subtree rooted at $P'$ in our analysis.

In view of this, we will use a pruning strategy (called \prunei) which prunes away all nodes $P'$ where $\card{\relevant{G'}{P'}}$ is small.
Using our new analysis of the diameter of the \JLS shortcut set, we will be able to say that the number of nodes $P'$ where $\card{\relevant{G'}{p'}}$ is large is much less than $n^{1/2 + O(1/\log k)}$, which yields our improvement.
See \Cref{subsec:prune-1} for details.

We close this subsection with the following crucial lemma from \cite{jambulapati2019parallel} which says that the size of balls around vertices are exponentially decreasing in $r$, the recursion level.

\begin{proposition}[Paraphrasing Lemma 4.1 from \cite{jambulapati2019parallel}]
\label{prop:jls-1}
    With probability at least $1 - O(n^{-10})$, the following event $\cE_{\ref{prop:jls-1}}$ holds.
    For all recursion levels $r$, for all level $r$ recursive instances $G'$,
    \begin{align*}
        &\card{\anc{G'}{v}} \le n k^{-r}
        \\
        &\card{\desc{G'}{v}} \le n k^{-r}
    \end{align*}
    for all $v \in V(G')$.
\end{proposition}

We omit rewriting a formal proof of \Cref{prop:jls-1}, but this follows from a simple induction on the recursion level: if $\card{\anc{G'}{v}} \le n k^{-r}$ at level $r-1$, then it is true also for level $r$ and, otherwise, one of the first $n k^{-r}$ ancestors of $v$ (say $p$) will be sampled as a pivot with high probability, precluding any of the later ancestors from being retained in $\anc{G'}{v}$ at level $r$ since they receive the label ``I have no relation to $p$'' while $v$ receives the label ``$p$ reaches me''.

%% file: prune-1.tex
\subsection{Bounding the Diameter Achieved From Using \prunei on \JLS}
\label{subsec:prune-1}

In this subsection, we first focus only on the reachability hopbound guarantee for \Cref{thm:seq-shortcut}.
We will later analyze the size of the shortcut sets and construction time in \Cref{subsec:shortcut-analysis}.

\begin{restatable}{theorem}{seqshorttradesideone}
\label{thm:seq-shortcut-tradeoff-side1}
    Let $\rho \in \IN$.
    The union of $\Theta(\log n)$ independent calls to
    \JLS \textcolor{\cprunei}{with \prunei} outputs an $n^{1/2 + o(1)}/\rho$-shortcut set with high probability.
\end{restatable}

To begin, \prunei uses the \emph{transitive closure} of subgraphs as a subroutine call.
We use $\TC(G)$ to denote all edges in the transitive closure of a digraph $G$; that is, $st \in \TC(G)$ if and only if $s$ can reach $t$ in $G$.
Note that $\TC(G)$ can be computed in $\Ot(\card{V(G)}^{\omega})$ time using repeated squaring of the adjacency matrix of $G$.

The very simple modification to \JLS is then described below (the text is mostly \JLS from the previous section, with the only substantial change being the addition of the \ding{47} line); we add edges from the transitive closure of the ball of each pivot if said balls are small.

\begin{tbox}
    \algname{\JLS \textcolor{\cprunei}{with \prunei}}\\
    \textbf{Global Parameters:} $k$ \textcolor{\cprunei}{and $\rho$} are global parameters to be fixed later. $n$ is the number of vertices in the base input graph (it thus remains fixed through all recursive calls).\\
    \textbf{Input:} A DAG $G = (V, E)$ and a recursion level $r$.\\
    \textbf{Output:} A shortcut set $H \subseteq V^2$.
    \begin{enumerate}
        \item Randomly and independently sample, with probability $p_r \gets \frac{100k^{r+1}\log n}{n}$, each $v \in V$ as a pivot.
            Let $S$ be the set of pivots.
        \item For each $p \in S$, compute $\overset{\textcolor{gray}{\textrm{reach }p}}{\anc{G}{p}}$, $\overset{\textcolor{gray}{p \textrm{ reaches}}}{\desc{G}{p}}$, \textcolor{\cprunei}{$\relevant{G}{p} = \anc{G}{p} \cup \desc{G}{p}$} and:
            \begin{itemize}
                \item Add $vp$ to $H$ for all $v \in \anc{G}{p}$.
                \item Add $pv$ to $H$ for all $v \in \desc{G}{p}$.
                \item Add $\vanc{p}$ label to all $v \in \anc{G}{p}$.
                \item Add $\vdesc{p}$ label to all $v \in \desc{G}{p}$.
                \item \label{alg:prune-1} \ding{47} \textcolor{\cprunei}{\prunei: If $\card{\relevant{G}{p}} \le \paren{k^2 \log^2 n} \rho^2$, add all edges in $\TC(G[\relevant{G}{p}])$ to $H$.}
            \end{itemize}
        \item $V_1, V_2, \ldots, V_t \gets$ Partition of all $v \in V$ such that $x, y \in V_i$ iff $x$ and $y$ have the exact same labels.
        \item Output $H \cup \paren{\bigcup_{i \in [t]} \JLS \textrm{ \textcolor{\cprunei}{with \prunei} on } G[V_i], r+1}$.
    \end{enumerate}
\end{tbox}

As before, we select an arbitrary path $P$ and count the nodes of its now pruned subproblem tree.
Below, we describe how \prunei allows us to prune the subproblem tree.

\begin{ttbox}
    \algname{\prunei on the Subproblem Tree}\\
    Consider a level $r$ node $P'$ in the original subproblem tree satisfying $\card{\relevant{G'}{P'}} \le \paren{k^2 \log^2 n} \rho^2$, and which also has children $P'_1, \ldots, P'_z$ arranged in the order so that $P' = P'_1 \ldots P'_z$.
    Let $P'_x$ be the last subpath that is touched by a call to $\TC(G'[\relevant{G'}{p}])$ for some level $r$ pivot $p$ that is a descendant of $P'$, and let $P'_y$ be the first subpath that is touched by a call to $\TC(G'[\relevant{G'}{q}])$ for some level $r$ pivot $q$ that is an ancestor of $P'$.
    We remove $P'_1, \ldots, P'_x$ and $P'_y, \ldots, P'_z$ (along with the subtrees rooted at them) from the subproblem tree of $P$, so that $P'$ now has children $P'_{x+1}, \ldots, P'_{y-1}$.
\end{ttbox}
In the above event, $P'_1, \ldots, P'_x$ (resp. $P'_y, \ldots, P'_z$) has been shortcut to $1$ hop from the call to $\TC(G'[\relevant{G'}{p}])$ (resp. $\TC(G'[\relevant{G'}{q}])$) from which a shortcut is added from the start of $P'_1$ to the end of $P'_x$ (resp. start of $P'_y$ to the end of $P'_z$).
The number of nodes in the pruned subproblem tree is thus, up to a constant factor, an upper bound on the number of hops $P$ is shortcut to.
Moving forward, we will call nodes $P'$ in subproblem $G'$ \emph{small} if $\card{\relevant{G'}{P'}} \le \rho^2$, and otherwise we call them \emph{large}.
There are three basic steps to count the number of nodes in the pruned subproblem tree:
\begin{itemize}
    \item We show that there is at most $n^{1/2 + O(1/\log k)}/\rho$ large nodes.
        See \Cref{subsubsec:p-1}.
    \item We show that the number of children each node has is $O(k \log n)$.
        We can use this to charge the maximal subtrees rooted at small nodes to their parent (a large node).
        See \Cref{subsubsec:p-2}.
    \item We show that for every small node, its subtree has at most $O(\log n)$ nodes.
        See \Cref{subsubsec:p-3}.
\end{itemize}
In all, letting $L$ be the number of large nodes, $T$ be an upper bound on the size of subtrees rooted at small nodes, and $\Delta$ upper bound the number of children each node has, then for $k = \Theta(\log n)$
\begin{align*}
    \paren{\textrm{\# of nodes}} \le L + L\Delta T = n^{1/2 + o(1)}/\rho.
\end{align*}
See \Cref{subsubsec:p-4}.

\subsubsection{There Can't be Many Large Nodes (Alternate Proof of the \JLS Diameter)}
\label{subsubsec:p-1}
We finally provide the proof of the $n^{1/2 + O(1/\log k)}$ diameter bound of \JLS.
To proceed, we will need the following key lemma as a blackbox (to keep this more self-contained, we will provide a brief proof sketch in \Cref{app:sketch-key-lemma}).
\begin{proposition}[Paraphrasing of Lemma 4.4 from \cite{jambulapati2019parallel}]
\label{prop:jls-2}
    Let $P'$ be an arbitrary path in $G'$, and $S$ be the set of pivots selected from $\relevant{G'}{P'}$.
    Suppose $S$ does not contain bridges and $S$ splits $P'$ into (possibly empty) subpaths $P'_1, \ldots P'_{\card{S}+1}$ belonging respectively to recursive instances $G'_1 \ldots G'_{\card{S}+1}$.
    Then
    \begin{align*}
        \expect{\sum_{i = 1}^{\card{S} + 1} \card{\relevant{G'_i}{P'_i}} \;\middle|\; \card{S}} \le \frac{2}{\card{S} + 1} \card{\relevant{G'}{P'}}.
    \end{align*}
\end{proposition}

While splitting a path $P'$ into $P'_1, P'_2, \ldots, P'_{\card{S}+1}$ precludes shortcutting $P'$ (hence $P$) to $o(\card{S})$ hops, \Cref{prop:jls-2} says that, on average, it becomes $\Omega(\card{S})$ times more likely to select pivots on $P'$ and thus shortcut its descendants in the subproblem tree to $O(1)$ hops each.
So even if the algorithm fails to resolve $P'$ by selecting a bridge, it makes progress towards resolving $P'_1, P'_2, \ldots, P'_{\card{S}+1}$ and it is this advantage that leads to the diameter bound of \cite{jambulapati2019parallel}.

Henceforth, we will use the notation $G_{P'}$ to refer to the recursive instance a subpath $P'$ belongs to.

\begin{lemma}
\label{lem:topdown-jls}
    For any $x$, let $X$ be the random variable counting the number of nodes $P'$ in the (unpruned) subproblem tree such that $\card{\relevant{G_{P'}}{P'}} \ge x$. Then:
    \begin{align*}
        \expect{X} \le \sqrt{\frac{n^{1 + O(1/\log k)}}{x}}.
    \end{align*}
\end{lemma}
\begin{proof}
    Let $T$ be the (unpruned) subproblem tree.
    \ \\
    \underline{Potential Function}:
    We use the potential function $\phi(P') = \sqrt{\card{\relevant{G_{P'}}{P'}}}$, and will show that $\expect{\sum_{\text{tree nodes } P'} \phi(P')} \le n^{1/2 + O(1/\log k)}$.
    If we can show this, we are done since nodes $P'$ with $\card{\relevant{G_{P'}}{P'}} \ge x$ contribute at least $\sqrt{x}$ each to the aforementioned sum; there can thus be at most $X \le n^{1/2 + O(1/\log k)}/\sqrt{x}$ such nodes in expectation.
    \\
    \underline{Local Step}: We first show that for any $P'$
    \begin{align*}
         \expect{\sum_{P'' \in \children{P'}} \phi(P'')} \le \sqrt{2}\expect{\phi(P')}.
         \tag{\ding{96}}
    \end{align*}
    The expectations in the following chain of inequalities are conditioned on the value of $\phi(P')$.
    \begin{align*}
        &
        \expect{\sum_{P'' \in \children{P'}} \phi(P'')}
        =
        \sum_{C \in [n]} \prob{\card{\children{P'}} = C
        }\expect{\sum_{P'' \in \children{P'}} \phi(P'') \;\middle|\; \card{\children{P'}} = C}
        \\
        &
        \le
        \sum_{C \in [n]} \prob{\card{\children{P'}} = C} \sqrt{C} \cdot \expect{\sqrt{\sum_{P'' \in \children{P'}} \card{\relevant{G_{P''}}{P''}}} \;\middle|\; \card{\children{P'}} = C}
        \tag{Cauchy-Schwarz}
        \\
        &
        \le
        \sum_{C \in [n]} \prob{\card{\children{P'}} = C} \sqrt{C \cdot \expect{\sum_{P'' \in \children{P'}} \card{\relevant{G_{P''}}{P''}} \;\middle|\; \card{\children{P'}} = C}}
        \tag{Jensen's Inequality}
        \\
        &
        \le
        \sum_{C \in [n]} \prob{\card{\children{P'}} = C} \sqrt{C \cdot \frac{2}{C} \card{\relevant{G'}{P'}}}
        \tag{Key \JLS Lemma: \Cref{prop:jls-2}}
        \\
        &
        =
        \sqrt{2} \phi(P').
    \end{align*}
    Using $\expect{\expect{X \;\middle|\; \phi(P')}} = \expect{X}$ on both sides gives $\expect{\sum_{P'' \in \children{P'}} \phi(P'')} \le \sqrt{2}\expect{\phi(P')}$, establishing \ding{96}.
    \\
    \underline{Summing the Pieces Up}:
    The proof of \Cref{lem:topdown-jls} is then easily completed by summing up over the nodes of $T$ by levels, and using induction on the level to compute $\paren{\# \textrm{ expected nodes in level }r} \le \sqrt{2}^r \expect{\phi(P)}$.
    Note that $T$ (deterministically) has $O(\log n / \log k)$ levels, since the sampling probability of being a pivot at recursion level $\Omega(\log n / \log k)$ is $1$ after which the algorithm halts.
    \begin{align*}
        \expect{\sum_{P' \in T} \phi(P')}
        &
        =
        \sum_{r \in [O(\log n / \log k)]} \expect{\sum_{P' \textrm{ in level } r \textrm{ of } T} \phi(P')}
        \tag{$T$ has $O(\log n / \log k)$ levels}
        \\
        &
        \le
        \sum_{r \in [O(\log n / \log k)]} \sqrt{2}^r \expect{\phi(P)}
        \tag{\ding{96}}
        \\
        &
        \le
        \sqrt{n} \sum_{r \in [O(\log n / \log k)]} \sqrt{2}^r 
        \tag{$\phi(P) \le \sqrt{n}$}
        \\
        &
        =
        n^{1/2 + O(1/\log k)}.
    \end{align*}
\end{proof}
As a special case when $x = 1$, we recover the diameter bound of \JLS from \Cref{lem:topdown-jls} since every node $P'$ must have $\card{\relevant{G_{P'}}{P'}} \ge 1$.
More importantly, for our proof, we will use $x = \rho^2$, the threshold which separates small nodes from large.
By \Cref{lem:topdown-jls}, there are no more than $n^{1/2 + o(1)}/\rho$ large nodes in expectation when $k = \Theta(\log n)$.

\subsubsection{Nodes Have Few Children}
\label{subsubsec:p-2}
\begin{observation}
\label{obs:few-children}
    With probability at least $1 - O(n^{-10})$, the event $\cE_{\ref{obs:few-children}}$ where every node in the subproblem tree has $O(k \log n)$ children holds.
\end{observation}
\begin{proof}
    We will condition on $\cE_{\ref{prop:jls-1}}$ holding, which occurs with probability $1 - O(n^{-10})$ by \Cref{prop:jls-1}.
    Let $P' = \langle v_0, v_1, \ldots, v_h \rangle$ be any level $r$ node in the subproblem tree, contained in some level $r$ recursive graph $G'$.
    Since $\relevant{G'}{P'} = \desc{G'}{v_0} \cup \anc{G'}{v_h}$ and by $\cE_{\ref{prop:jls-1}}$, it follows that $\card{\relevant{G'}{P'}} \le 2nk^{-r}$.

    Let $X_{P'}$ be the number of pivots chosen from $\relevant{G'}{P'}$ at level $r$.
    Since the sampling rate at level $r$ is $100k^{r+1}\log n / n$, we have $\expect{X_{P'}} \le 100k\log n$.
    By a Chernoff bound, $\prob{X_{P'} > 200k\log n} \le O(n^{-12})$.
    Therefore, $\prob{\bigcup_{\textrm{nodes }P'} \set{X_{P'} > 200k\log n}} \le O(n^{-10})$.
    This implies, using \Cref{obs:s-splitting} which bounds the number of pieces $P'$ is split to by the number of pivots, that $\prob{\cE_{\ref{obs:few-children}}} \ge 1 - O(n^{-10})$.
\end{proof}

\subsubsection{Subtrees Rooted at Small Nodes are Heavily Pruned}
\label{subsubsec:p-3}
\begin{lemma}
\label{lem:small-subtree}
    Let $P'$ be a small node in the (pruned) subproblem tree, and let $T_{P'}$ be the subtree rooted at it.
    \begin{align*}
        \expect{\card{V(T_{P'})}} = O(\log n).
    \end{align*}
\end{lemma}
\begin{proof}
    Assume $\cE_{\ref{prop:jls-1}}$ holds.\footnote{The contribution from $\cE_{\ref{prop:jls-1}}$ not holding is negligible since $\prob{\textrm{not }\cE_{\ref{prop:jls-1}}}\expect{\card{V(T_{P'})} \;\middle|\; \textrm{not }\cE_{\ref{prop:jls-1}}} < 1$. This follows from \Cref{prop:jls-1} and $\card{V(T_{P'})} \le n$.}
    For any small node $P''$ at level $r$, we will call it \emph{bad} if there is some $v \in \relevant{G_{P''}}{P''}$ such that $\card{\relevant{G_{P''}}{v}} > (k^2 \log^2 n) \rho^2$; otherwise $P''$ is \emph{good}.
    \\
    \underline{Property of good nodes}:
    Notice that if $P''$ is good, then it has at most $1$ child in the pruned subproblem tree since we run $\TC(G[\relevant{G_{P''}}{p}])$ for any pivot $p \in \relevant{G_{P''}}{P''}$.
    \\
    \underline{Property of bad nodes}:
    If $P''$ is bad, there is some $v \in \relevant{G_{P''}}{P''}$ such that $\card{\relevant{G_{P''}}{v}} > (k^2 \log^2 n) \rho^2$.
    Using
    \begin{align*}
        (k^2\log^2 n) \card{\relevant{G_{P''}}{P''}} \underset{P''\textrm{ small}}{\le} (k^2 \log^2 n) \rho^2 \underset{P''\textrm{ bad}}{<} \card{\relevant{G_{P''}}{v}} \underset{\cE_{\ref{prop:jls-1}}}{\le} nk^{-r},
    \end{align*}
    we get $\card{\relevant{G_{P''}}{P''}} < nk^{-r-2}\log^{-2}n$.
    Consequently, the expected number of pivots selected from $\relevant{G_{P''}}{P''}$ is at most $O(k^{-1}\log^{-1}n)$.
    Letting $X_{P''}$ be the number of children $P''$ has, and using \Cref{obs:s-splitting} with the expected number of pivots, $\expect{X_{P''}} = 1 + O(k^{-1}\log^{-1}n)$.
    \\
    \underline{Total number of nodes}:
    Let $Z_r$ for $r \le \log n / \log k$ be the number of nodes in level $r$ of $T_P$.
    Conditioning on $Z_r$, the above bound $\expect{X_{P''} \;\middle|\; Z_r} = 1 + O(k^{-1}\log^{-1}n)$ still holds and hence it follows that
    \begin{align*}
        \expect{Z_{r+1}}
        &
        =
        \expectt{Z_r}{\expectt{Z_{r+1}}{Z_{r+1} \;\middle|\; Z_r}}
        \tag{Law of iterated expectations}
        \\
        &
        =
        \expectt{Z_r}{\expectt{Z_{r+1}}{\sum_{i \in [Z_r]} X_{P''_i} \;\middle|\; Z_r}}
        \tag{where $P''_i$ is the $i$th node in level $r$}
        \\
        &
        =
        \expect{\paren{1 + O(k^{-1}\log^{-1}n)} \cdot Z_r}
        \tag{$\expect{X_{P''} \;\middle|\; Z_r} = 1 + O(k^{-1}\log^{-1}n)$}
        \\
        &
        =
        \paren{1 + O(k^{-1}\log^{-1}n)}^{\log n / \log k}
        \tag{$Z_0 = 1$ and $r \le \log n / \log k$}
        \\
        &
        =
        O(1).
    \end{align*}
    We conclude that $\expect{\card{V(T_{P'})}} = \expect{\sum_{r \in [\log n / \log k]} Z_r} \le \frac{\log n}{\log k} \cdot O(1) = O(\log n)$.
\end{proof}

\subsubsection{Putting the Pieces Together}
\label{subsubsec:p-4}
We now have all the components to prove \Cref{thm:seq-shortcut-tradeoff-side1}.
\begin{proof}[Proof of \Cref{thm:seq-shortcut-tradeoff-side1}]
    Let $P$ be any path in $G$, and $T_P$ be its subproblem tree, and $k = \Theta(\log n)$.
    We will show later that $\expect{\card{V(T_P)}} \le n^{1/2 + o(1)} / \rho$.
    Then, by Markov's inequality $\prob{\card{V(T_P)} \le n^{1/2 + o(1)} / \rho} > 1/2$ so that repeating \JLS \textcolor{\cprunei}{with \prunei} $\Theta(\log n)$ times shortcuts $P$ with probability $\Omega(n^{-3})$.
    Union bounding over $O(n^2)$ paths (one chosen for each pair $s,t \in V$), the algorithm shortcuts all paths with probability $\Omega(n^{-1})$.
    Let us hence return to showing that $\expect{\card{V(T_P)}} \le n^{1/2 + o(1)} / \rho$.

    Assume $\cE_{\ref{obs:few-children}}$ holds\footnote{The contribution from $\cE_{\ref{obs:few-children}}$ not holding is negligible since $\prob{\textrm{not }\cE_{\ref{obs:few-children}}}\expect{\card{V(T_P)} \;\middle|\; \textrm{not }\cE_{\ref{obs:few-children}}} < 1$. This follows from \Cref{obs:few-children} and $\card{V(T_P)} \le n$.}.
    Let $L$ be the number of large nodes in $T_P$.
    We may upper bound $\card{V(T_P)}$ by
    \begin{align*}
        \expect{\card{V(T_P)}}
        &
        \le
        \expect{L + \sum_{\textrm{large } P'} \sum_{\substack{P'' \in \children{P'}\\: P'' \textrm{ small}}} \card{V(T_{P''})}}
        \\
        &
        \le
        \expect{L} + \expect{L \cdot \underset{\ref{obs:few-children}}{O(\log^2 n)} \underset{\ref{lem:small-subtree}}{O(\log n)}}
        \tag{\Cref{obs:few-children} and \Cref{lem:small-subtree}}
        \\
        &
        =
        \paren{1 + O(\log^3 n)}\expect{L} = n^{1/2 + o(1)}/\rho.
        \tag{\Cref{lem:topdown-jls}}
    \end{align*}
\end{proof}

%% file: shortcut-analysis.tex
\subsection{Remaining Analysis of the Shortcut Set}
\label{subsec:shortcut-analysis}

Here we finally prove \Cref{thm:seq-shortcut} in full.

\seqshort*
\begin{proof}
    With \Cref{thm:seq-shortcut-tradeoff-side1}, it only remains to bound the running time of \JLS \textcolor{\cprunei}{with \prunei} and the size of the shortcut set it produces.
    We bound the contribution from \prunei since the contribution from \JLS is taken care of by \Cref{prop:jls-main}.
    Set $k = \Theta(\log n)$.
    We will condition on the following event:
    \begin{align*}
        \set{\forall r \textrm{ there are } \Ot(k^{r+1}) \textrm{ pivots sampled at level }r} \cap \cE_{\ref{prop:jls-1}},
    \end{align*}
    which occurs with probability at least $1 - O(n^{-10})$ by a Chernoff bound, \Cref{prop:jls-1}, and a union bound.
    \\
    \underline{Time}:
    \JLS takes $\Ot(m)$ time by \Cref{prop:jls-main}.
    We next show that the transitive closure calls, from \prunei, takes $\Ot(n\rho^{2\omega - 2})$ time.
    Recall that each level $r$ transitive closure call is made in a subgraph induced on $\relevant{G}{p}$ for each level $r$ pivot $p$, so long as $\card{\relevant{G}{p}} = \Ot(\rho^2)$.
    This takes $\Ot(\card{\relevant{G}{p}}^\omega)$ time per call, using matrix multiplication and repeated squaring of the adjacency matrix.
    We will break these calls into two cases, based on the level of recursion $r$.
    \begin{itemize}
        \item Case 1: $k^r < n/\rho^2$.
            Since there are $\Ot(k^r)$ pivots sampled at recursion level $r$ and below, the time at that level is, up to polylogarithmic factors, $k^r (\rho^2)^{\omega} < n\rho^{2\omega - 2}$.
            Since there are $O(\log n / \log\log n)$ levels, the bound follows.
        \item Case 2: $k^r \ge n/\rho^2$.
            Since there are $\Ot(k^r)$ pivots sampled at recursion level $r$ and $\cE_{\ref{prop:jls-1}}$ holds, the time at that level is, up to polylogarithmic factors,
            \begin{align*}
                k^r \paren{\frac{n}{k^r}}^\omega = \frac{n^{\omega}}{k^{r(\omega - 1)}} \le \frac{n^{\omega}}{(n/\rho^2)^{\omega - 1}} = n\rho^{2\omega -2}.
            \end{align*}
            Since there are $O(\log n / \log\log n)$ levels, the bound follows.
    \end{itemize}
    \underline{Size}:
    The argument for this is similar to the time bound.
    \begin{itemize}
        \item Case 1: $k^r < n/\rho^2$.
            Since there are $\Ot(k^r)$ pivots sampled at recursion level $r$, the number of shortcuts added at that level is, up to polylogarithmic factors, $k^r (\rho^2)^2 < n\rho^2$.
            Since there are $O(\log n / \log\log n)$ levels, the bound follows.
        \item Case 2: $k^r \ge n/\rho^2$.
            Since there are $\Ot(k^r)$ pivots sampled at recursion level $r$ and $\cE_{\ref{prop:jls-1}}$ holds, the number of shortcuts added at that level is, up to polylogarithmic factors,
            \begin{align*}
                k^r \paren{\frac{n}{k^r}}^2 = \frac{n^2}{k^r} \le \frac{n^2}{n/\rho^2} = n\rho^2.
            \end{align*}
            Since there are $O(\log n / \log\log n)$ levels, the bound follows.
    \end{itemize}
    \underline{Main case ($\Ot(m)$ size shortcut set)}:
    Set $\rho = \paren{m/n}^{1/(2\omega - 2)}$.
\end{proof}

%% file: reach.tex
\section{Parallel Reachability}
\label{sec:reach}

In this short section, we parallelize \JLS \textcolor{\cprunei}{with \prunei}, the shortcut set construction shown in \Cref{sec:seq-shortcut}, proving \Cref{thm:par-shortcut}.
We then use the shortcut set to prove \Cref{thm:par-reach}, our main result for parallel reachability.

We use the following result from \cite{blackbox} to reduce the construction of shortcut sets on digraphs to that on so-called shallow digraphs.

\begin{proposition}[Paraphrasing Corollary 3.2 from \cite{blackbox}]
\label{prop:shallow-shortcut}
    Suppose $\lambda > c_0 \log^3 n$ and $\alpha < 1/(c_0\log^2_\lambda n)$ where $c_0$ is a sufficiently large constant.
    
    Suppose there is a parallel algorithm  $\cA_0$ that, given a digraph $G_{0}$ with $n$ vertices and $m_0$ edges and reachability hopbound $h_0 = \lambda \beta$, returns a $\beta$-shortcut set of size $\alpha m_0 + f(n)$.
    
    Then there is a randomized parallel algorithm $\cA$ that, given a digraph $G$ with $n$ vertices and $m$ edges, returns a $\beta$-shortcut set of size $S(m) = \Ot(\alpha m + f(n))$.
    $\cA$ makes a polylogarithmic number of sequential calls to $\cA_0$ on digraphs with at most $S(m) + m$ edges, and takes an additional $\Ot(m)$ work and $\Ot(h_0)$ span.
\end{proposition}

We are now ready to parallelize \JLS \textcolor{\cprunei}{with \prunei}.

\begin{restatable}[Parallel Near Linear Work Shortcut Set Construction]{theorem}{parshort}
\label{thm:par-shortcut}
    Let $\omeganc$ be the fast matrix multiplication exponent in the work for parallel algorithms using polylogarithmic (in $n$) span.
    There is an $\Ot(m)$ work
    and $\sqrt[2\omeganc-2]{n^{\omeganc+o(1)}/m}$ span
    randomized parallel algorithm that, given an unweighted digraph $G$, outputs with high probability a
    $\sqrt[2\omeganc-2]{n^{\omeganc+o(1)}/m}$-shortcut set $H$ with size 
    $\Ot\paren{m}$.
\end{restatable}
\begin{proof}
    \ \\
    \underline{Implementing $\cA_0$ (Part 1) --- Parallel reduction to shallow DAGs\footnote{This part may be used for any shortcut set algorithm, hence one may assume that $\cA_0$ is given a shallow DAG (as opposed to a shallow digraph).}}:
    Observe that we can use the algorithm of \cite{schudy2008finding} with a parallel BFS oracle to find the SCCs of $G_0$ in $\Ot(m)$ work and $\Ot(h_0)$ span.
    To see this, note that the algorithm of \cite{schudy2008finding} recurses into graphs induced on intervals of the topological order of the SCCs of $G_0$, hence each recursive instance maintains a $h_0$ reachability hopbound.
    
    By adding a bidirected star in each SCC of $G_0$ to our shortcut set, it then suffices to construct a $\beta/2$-shortcut set for the DAG where the SCCs of $G_0$ are contracted.
    This incurs at most a factor of two dilation to give a $\beta$-shortcut set for $G_0$.
    We henceforth assume $G_0$ is a DAG with reachability hopbound $h_0$.
    \\
    \underline{Implementing $\cA_0$ (Part 2) --- Parallel \JLS \textcolor{\cprunei}{with \prunei} on shallow DAGs}:
    Observe that for any $h_0 \gg n^{1/2 + o(1)}/\rho$, \JLS \textcolor{\cprunei}{with \prunei} runs in work $\Ot\paren{m + n\rho^{2\omeganc-2}}$ and span $\Ot(h_0)$ for DAGs with reachability hopbound $h_0$.
    
    To see this, first note that the $h_0$ reachability hopbound is preserved in all recursive subgraphs: if an arbitrary pair $u \preceq v$ are sent to the same recursive instance $G'$, then every vertex between $u$ and $v$ have the same reachability relations (to the pivots) as $u,v$ and are thus also sent to $G'$.
    In particular, the $h_0$ hop path from $u$ to $v$ is sent to $G'$.
    The reachability relations, which are computed by parallel BFS calls, are therefore done in $\Ot(\card{E(G')})$ work and $\Ot(h_0)$ span in each recursive instance $G'$.
    Next, note that each transitive closure call is done in $\Ot(\rho^{2\omeganc})$ work and polylogarithmic in $n$ span.
    \\
    \underline{Setting parameters}:
    Fix $\rho = \paren{m_0/n}^{1/(2\omeganc - 2)}/\textrm{polylog}(n)$ for a sufficiently large power in the polylogarithm term.
    Let the output of \JLS \textcolor{\cprunei}{with \prunei} on shallow DAGs be a $\beta$-shortcut set $H$ with $\beta = n^{1/2 + o(1)}/\rho = \sqrt[2\omeganc-2]{n^{\omeganc+o(1)}/m_0}$.
    Note that for this setting of $\rho$, we have $\card{H} = \Ot(n\rho^2) = \alpha m_0$ where $\alpha < 1/(c_0\log^2_\lambda n)$.
    Then, pulling back to shallow digraphs, note that $f(n)$ accounts for the edges added by the bidirected stars and hence $f(n) \le n$ and $S(m) = \Ot(m)$.
    \\
    \underline{Putting things together}:
    Invoking \Cref{prop:shallow-shortcut}, $\cA$ outputs with high probability a
    $\sqrt[2\omeganc-2]{n^{\omeganc+o(1)}/m}$
    -shortcut set with size $\Ot(m)$
    in $\Ot(m)$ work and $\sqrt[2\omeganc-2]{n^{\omeganc+o(1)}/m}$ span.
\end{proof}

\parreach*
\begin{proof}
    This follows from using \Cref{thm:par-shortcut} to extract a $\sqrt[2\omeganc-2]{n^{\omeganc+o(1)}/m}$
    -shortcut set $H$ of $G$ with size $\Ot(m)$
    in $\Ot(m)$ work and $\sqrt[2\omeganc-2]{n^{\omeganc+o(1)}/m}$ span.
    Then, run a parallel BFS on $G \cup H$ from $s$ in $\Ot(m)$ work and $\sqrt[2\omeganc-2]{n^{\omeganc+o(1)}/m}$ span.
\end{proof}

%% file: hopset.tex
\section{Sequential Hopset Construction}
\label{sec:seq-hopset}

In this section we prove \Cref{thm:seq-hopset}.
\Cref{subsec:cfr} gives a high level overview of the \CFR hopset.
\Cref{subsec:logical} goes into further detail, summarizing aspects of the analysis in \cite{cao2020efficient} that are relevant to us (for full detail, refer to \cite{cao2020efficient}); there are also some general lemmas related to the \CFR analysis proved in this section.
We then show how to improve the hopset bounds via a small change to \CFR in \Cref{subsec:prune-2} and finish things up.

\input{cfr}
\input{logical}
\input{prune-2}

%% file: cfr.tex
\subsection{The \CFR Hopset}
\label{subsec:cfr}

In this section, we describe the hopset construction (\CFR) of \cite{cao2020efficient}.

\paragraph{High Level Idea of \CFR. }
It is instructive to first note that \JLS with weights given to shortcut edges does not work for computing a $(\beta, \eps)$-hopset since bridges are not necessarily good pivots.
For example, consider our analysis for shortcut sets where we pick an aribtrary path $P$ and examine its subproblem tree.
Consider a node in this tree --- say, a subpath $P'$ from $u$ to $v$ --- and what happens if we select a bridge $p$.
In \JLS, $P'$ would then be a leaf of this tree and we would shortcut $P'$ by taking the edges $up$ and $pv$.
But this may accumulate too high an error since it is possible that $\dist(u,p) + \dist(p,v) \gg \dist(u,v)$.

The main idea of \cite{cao2020efficient}, then, is to make sure that we can indeed shortcut through bridges without paying too much in the error.
Consider the following idealized scenario for a fixed path $P$ which we wish to shortcut:
\begin{itemize}
    \item Every level $r$ subpath in the subproblem tree has length $D_r$, which is decreasing in $r$.
    \item Instead of running SSSP to/from each pivot, we run a $10 D_r$ distance-limited SSSP.
\end{itemize}
Suppose $P''$ is a level $r$ path from $u''$ to $v''$ for which a bridge $p$ is selected as a pivot (i.e. in a $10D_r$ distance limited to / from $p$, there are some vertices in $P''$ that reach / are reached by $p$).
Consider the level $r-L$ ancestor $P'$ of $P''$ from $u'$ to $v'$, with $L$ chosen to satisfy $10D_r \approx\eps D_{r-L}/2$.
Instead of using the shortcuts $u''p$ and $pv''$, which do not approximate $P''$ well, we can use the shortcuts $u'p$ and $pv'$ to approximate $P'$ since $\dist(u',p) + \dist(p,v') \le D_{r-L} + 20 D_r \le (1+\eps)D_{r-L}$.
Put otherwise, $p$ being a bridge for the level $r$ subpath $P''$ certifies that $p$ can be used to shortcut the level $r-L$ ancestor subpath $P'$ without too much error.

At any rate, the \CFR construction can be described at a high level as follows.
Let $W$ be the largest edge weight, which we assume to be polynomially bounded.
Firstly, $\log(nW)$ many path distances are guessed and, for a guess $D$, paths with length roughly $D$ will be shortcutted by the hopset to $n^{1/2 + o(1)}$ hops.
The hopset is constructed by an algorithm described recursively.
At recursion level $r$, roughly $k^r$ pivots are selected at random, like \JLS.
While there is a minor but technical issue where we are no longer able to assume that the input is a DAG, there are two key differences here with \JLS to highlight:
\begin{itemize}
    \item Instead of reachability relations like \JLS, a distance-limited search is computed and therefrom are the relations derived.
        The distance-limit at level $r$ is roughly $D_r$ multiplied with lower order terms
        where the distance-limits are defined to be shrinking: $\sqrt{k} D_r \le D_{r-1}$ and $D_0 = D$.
        Also, in order to get similar guarantees to \JLS, the searches will really be done up to distance $D_r$ multiplied by a random scalar between $[\eta_{\textrm{min}}, 2(\eta_{\textrm{min}}+1)]$.
        Here, the terms $k$ and $\eta_{\textrm{min}}$ are to be determined later.
    \item In \JLS, randomly sampled vertices actually play two conceptually different roles at the same time: \emph{shortcutters} and \emph{pivots}.
        A vertex is a shortcutter when edges are added to/from it.
        A vertex is a pivot when it is used to determine recursive subproblems.
        \CFR decouples these roles, where level $r$ pivots are shortcutters at level $r - L$, with $L$ to be determined later.
        That is, distance-limited searches to/from shortcutters are computed at the earlier recursion level $r-L$ with distance-limit roughly $D_{r-L} \gg D_r$ and it is from here that edges are added to the hopset.
        These same vertices will send labels out at level $r$, where the distance-limit is $D_r$, to determine the level $r+1$ subproblems.
        When a pivot is sampled at level $r$, it being a bridge would certify that we can shortcut through it at the earlier level $r-L$.
\end{itemize}

Since the searches are distance-limited, a problem arises when considering a path $P$ in the analysis.
When some pivots $S$ are selected at level $r$ so that for all $s \in S$ either $s \preceq_{D_r} P$ or $P \preceq_{D_r} s$, it does not mean that $P$ will be split into at most $O(\card{S})$ subpaths (i.e. the analogue of \Cref{obs:s-splitting} does not hold) and, moreover, a recursive subproblem could contain non-contiguous subpaths (see \Cref{fig:broken-p}).

\begin{figure}[h]
    \centering
    \includegraphics[scale=0.6]{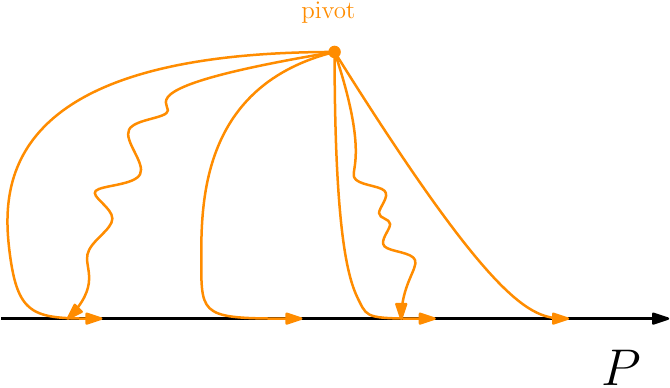}
    \caption{
        An example where a single (ancestor) pivot chosen.
        $P$ is split across two subproblems, each containing a collection of disjoint subpaths of $P$.
        In \JLS, on the other hand, $P$ would have just been split into a left half (before the first point reached by the pivot) and a right half.
        This issue occurs because searches are distance-limited in \CFR.
    }
    \label{fig:broken-p}
\end{figure}
To circumvent this, \CFR also recurses into \emph{fringe} subproblems (with the usual subproblem being called a \emph{core} subproblem) which enable us, with some care taken, to get an analogue of \Cref{obs:s-splitting}.
These fringe subproblems are comprised of vertices at the boundaries of distance-limited searches from pivots.
We have omitted some detail (unnecessary to understand our proofs) from the above description, but the algorithm is described completely below (with helpful comments in blue text).

To fully realize the main insight above (where bridges certify good shortcuts) we are left with the question of how to ensure level $r$ subproblems have length no more than $D_r$.
We will defer this to \Cref{subsec:logical}, as that property is maintained in the analysis and not algorithmically.
The key points to take away, before we move on to the formal description of \CFR, are: (i) Searches are now distance-limited; (ii) Shortcutting and pivoting have been decoupled; (iii) There are two types of subproblems, namely, core and fringe.

\paragraph{The Hopset Construction. }
There are global parameters floating around in their algorithm, which will eventually be finalized to:
\begin{itemize}
    \item $k = \Theta(\log n)$.
    \item $L = 15 - 2 \log_k \eps$.
    \item $\lambda = 100$.
    \item $c$ satisfies $k^c = \Theta\paren{\frac{\lambda^L k^{(L-1)/2}}{\log^3 n}}$.
    \item $\eta_{\textrm{min}} = 16\lambda^2 k^2 \log^2 n - 1$ and $\eta_{\textrm{max}} = 2(\eta_{\textrm{min}}+1)$.
\end{itemize}
We call this the \emph{final parameter setting}.
The details of these are not so important except that they are the settings in \cite{cao2020efficient} that yield a $(n^{1/2 + o(1)},\eps)$-hopset with size $\Ot(n)$ in $\Ot(m/\eps^2)$ time and this is the hopset we build off of.

\begin{tbox}
    \algname{Outer Shell of \CFR}\\
    \textbf{Global Parameters:} $k$, $L$, $\lambda$, and $c$ are global parameters to be fixed later.\\
    \textbf{Input:} A weighted directed graph $G = (V, E)$.\\
    \textbf{Output:} A weighted directed approximate hopset $H \subseteq V^2$.
    \begin{enumerate}
        \item Repeat $\Theta\paren{\log n}$ times: \qquad\qquad \textcolor{blue}{\texttt{(boosting to succeed whp)}}
            \begin{enumerate}
                \item For $j \in [-1, \log(nW)]$: \qquad\qquad \textcolor{blue}{\texttt{(guessing a distance $\approx 2^j$)}}
                    \begin{enumerate}
                        \item For each $v \in V$, set $\ell(v) \gets$ the smallest $r \in [0,\log_k n]$ such that a fresh $\textrm{Bernoulli}\paren{\frac{\lambda k^{r+1}\log n}{n}}$ succeeds.
                        \item For each $v \in V$ with $\ell(v) \le L$:
                            \begin{itemize}
                                \item Add $uv$ with weight $\dist_G(u,v)$ to $H$ for all $u \in \danc{G}{v}{2^{j+1}}$.\footnote{Recall the definition of $\danc{G}{v}{d}$ from \Cref{sec:preliminaries}.}
                                \item Add $vu$ with weight $\dist_G(v,u)$ to $H$ for all $u \in \ddesc{G}{v}{2^{j+1}}$.
                            \end{itemize}
                        \item $H \gets H\ \cup$ \CFR on $G, r = 0, D=2^j/k^c$.
                    \end{enumerate}
            \end{enumerate}
        \item Output $H$.
    \end{enumerate}
\end{tbox}

\begin{tbox}
    \algname{\CFR}\\
    \textbf{Global Parameters:} $k$, $L$, and $\lambda$ are global parameters. $\ell(v)$ is the recursion level at which $v$ becomes a pivot. $n$ is the number of vertices in the base input graph (it thus remains fixed through all recursive calls)\\
    \textbf{Input:} A weighted directed graph $G = (V, E)$, a recursion level $r$, and a guessed distance $D$.\\
    \textbf{Output:} A weighted directed approximate hopset $H \subseteq V^2$.
    \begin{enumerate}
        \item $D_r \gets \frac{D}{\lambda^r k^{r/2}}$.
        \item For each $v \in V$ with $\ell(v) = r + L$:
            \begin{itemize}
                \item Add $uv$ with weight $\dist_G(u,v)$ to $H$ for all $u \in \danc{G}{v}{32 \lambda^2 k^2 D_r \log^2 n}$.
                \item Add $vu$ with weight $\dist_G(v,u)$ to $H$ for all $u \in \ddesc{G}{v}{32 \lambda^2 k^2 D_r \log^2 n}$.
            \end{itemize}
        \item For each $v \in V$ with $\ell(v) = r$:
            \begin{enumerate}
                \item Sample $\sigma_v \in [1, 4\lambda^2 k \log^2 n]$ uniformly at random.
                \item \textcolor{blue}{\texttt{(minimizing fringe set size)}}\\
                    Choose $\eta_v \in 4k[(\sigma_v - 1), \sigma_v] + 16 \lambda^2 k^2 \log^2 n$ such that $\card{\drelevant{G}{v}{(\eta_v+1)} \setminus \drelevant{G}{v}{(\eta_v-1)}}$ is minimized.
                \item \textcolor{blue}{\texttt{(fringe set)}}\\
                    $V_v^{\textrm{fringe}} \gets \drelevant{G}{v}{(\eta_v+1)D_r} \setminus \drelevant{G}{v}{(\eta_v-1)D_r}$.
                \item \textcolor{blue}{\texttt{(solve fringe problem)}}\\
                    $H \gets H \cup \CFR \textrm{ on } G[V_v^{\textrm{fringe}}], r+1, D$.
                \item Add $\vanc{v}$ label to all $u \in \danc{G}{v}{\eta_v D_r} \setminus \ddesc{G}{v}{\eta_v D_r}$.
                \item Add $\vdesc{v}$ label to all $u \in \ddesc{G}{v}{\eta_v D_r}  \setminus \danc{G}{v}{\eta_v D_r}$.
                \item Add \ding{55} label to all $u \in \danc{G}{v}{\eta_v D_r} \cap \ddesc{G}{v}{\eta_v D_r}$.
            \end{enumerate}
        \item $V_1, V_2, \ldots, V_t \gets$ Partition of all $v \in V$ with no \ding{55} label, such that $x, y \in V_i$ iff $x$ and $y$ have the exact same labels.
            \qquad\qquad \textcolor{blue}{\texttt{(core subproblems)}}
        \item Output $H \cup \paren{\bigcup_{i \in [t]} \CFR \textrm{ on } G[V_i], r+1, D}$.
            \qquad\qquad \textcolor{blue}{\texttt{(solve core subproblems)}}
    \end{enumerate}
\end{tbox}

\begin{proposition}[Paraphrasing of Theorem 14 from \cite{cao2020efficient}]
\label{prop:cfr-main}
    Under the final parameter setting,
    the \CFR construction runs in time $\Ot(m/\eps^2)$ and produces a $(n^{1/2+o(1)}, \eps)$-hopset of size $\Ot(n/\eps^2)$ with high probability.
\end{proposition}

The next two statements from \cite{cao2020efficient} will also be useful in our analysis.
The first is an analogue \Cref{prop:jls-1} but in a distance-limited sense, and the second says that fringe subproblems are small in expectation.

\begin{proposition}[Paraphrasing of Lemma 3 from \cite{cao2020efficient}]
\label{prop:cfr-1}
    With probability at least $1 - O(n^{-10})$, the following event $\cE_{\ref{prop:cfr-1}}$ holds.
    For all recursion levels $r$, for all level $r$ recursive instances $G'$,
    \begin{align*}
        &\card{\danc{G'}{v}{\eta_{\textrm{max}}D_r}} \le n k^{-r}
        \\
        &\card{\ddesc{G'}{v}{\eta_{\textrm{max}}D_r}} \le n k^{-r}
    \end{align*}
    for all $v \in V(G')$.
\end{proposition}

\begin{proposition}[Paraphrasing of Lemma 4 from \cite{cao2020efficient}]
\label{prop:cfr-2}
    For any $v$ at any recursion level $r$,
    \begin{align*}
        \expect{\card{V_v^{\textrm{fringe}}}} \le \frac{1}{4 \lambda k \log n}.
    \end{align*}
\end{proposition}

%% file: logical.tex
\subsection{The Logical Subproblem Tree}
\label{subsec:logical}

To analyze the hopset bounds, \cite{cao2020efficient} fixes an arbitrary shortest path $P$ with length $D$, and focuses on a run of \CFR which guesses a distance within a factor $2$ of $D$.
Since fringe subproblems may overlap with each other and even core subproblems, defining the correct subproblem tree for $P$ to analyze is more delicate (whereas with \JLS, defining it was unambiguous).
To make matters even more complicated, since the searches from pivots are distance-limited, $P$ will not necessarily be split into contiguous subpaths in the next level of recursion (recall \Cref{fig:broken-p}).

As an answer to these issues, \cite{cao2020efficient} defines the \emph{logical} subproblem tree defined thusly.
\\
\underline{What a node is}:
A level $r$ \emph{logical subproblem} (i.e. node) is associated with a subpath $P'$ in a level $r$ recursive instance $G$, where we further maintain that $\length(P') \le D_r = D / (\lambda \sqrt{k})^r$.
\\
\underline{The children of a node}: 
The level $r$ logical subpath $P'$ will be split into \emph{actual} continguous subpaths $P'_1, \ldots, P'_z$ contained in distinct level $r+1$ recursive instances $G'_1, \ldots, G'_z$ (we will show how to choose which instances, fringe or core, later) and, importantly, each actual subpath $P'_i$ will be further split into $\lceil \length(P'_i) / D_{r+1} \rceil$ logical subpaths $P'_{i,1}, \ldots, P'_{i,\lceil \length(P'_i) / D_{r+1} \rceil}$ each of length at most $D_{r+1}$ and such that $P'_i = P'_{i,1} \ldots P'_{i,\lceil \length(P'_i) / D_r \rceil}$.
These logical subpaths $P'_{i,j}$ are the children of the logical\footnote{There are now four types of subproblems: $\set{\textrm{core, fringe}} \times \set{\textrm{logical, actual}}$. The logical subproblem tree only has nodes corresponding to logical subproblems/subpaths.} subpath $P'$.
For all $i$, we will call (the logical subpath) $P'_{i,1}$ the \emph{main copy} of (the actual subpath) $P'_i$, and $P'_{i,j}$ for $j > 1$ the \emph{secondary copies}.
\\
\underline{Choosing how to split a logical subproblem $P'$}:
We will split a level $r$ node $P'$ contained in recursive instance $G'$ in the following way (described also in the proof of Lemma 8 in \cite{cao2020efficient}).
Let $$S = \set{v: \ell(v) = r \textrm{ and } v \in \drelevant{G'}{P'}{\eta_v D_r}},$$ that is, $S$ is the set of level $r$ pivots $v$ which are $\eta_v D_r$-related to $P'$.
\begin{itemize}
    \item If $S$ is empty, a level $r+1$ copy of $P'$ is made the level $r$ copy's child (or children, if split into more than $1$ logical subproblem to maintain a length of at most $D_{r+1}$ at level $r+1$).
    \item If $S$ contains an $\eta_v D_r$-bridge, then $P'$ is a leaf.
    \item Otherwise, for each $\eta_v D_r$-ancestor $v \in S$, let $v_f$ be the first vertex in $P'$ such that $v \preceq_{\eta_v D_r} v_f$ and $v_c$ be the first vertex in $P'$ such that $v \preceq_{(\eta_v-1) D_r} v_f$.
    Then observe that
    \begin{itemize}
        \item $\langle v_0, \ldots, v_{f-1} \rangle$ are $\eta_v D_r$-unrelated to $P'$.
        \item $\langle v_f, \ldots, v_{c-1} \rangle$ are contained in $V_v^{\textrm{fringe}}$.
        \item $\langle v_c, \ldots \rangle$ are all $\eta_v D_r$-reached by $v$.
    \end{itemize}
    We similarly define, for each $\eta_v D_r$-descendant $v \in S$ the vertex $v_f$ to be the last vertex in $P'$ such that $v_f \preceq_{\eta_v D_r} v$ and $v_c$ to be the last vertex in $P'$ such that $v_c \preceq_{(\eta_v-1) D_r} v$.
    Now
    \begin{itemize}
        \item $\langle v_{f+1}, \ldots \rangle$ are $\eta_v D_r$-unrelated to $P'$.
        \item $\langle v_{c+1}, \ldots, v_f \rangle$ are contained in $V_v^{\textrm{fringe}}$.
        \item $\langle v_0, \ldots, v_c \rangle$ all $\eta_v D_r$-reach $v$.
    \end{itemize}
        We select up to $2\card{S} + 1$ subpaths of $P'$ by the following procedure:
        \begin{tbox}
            \algname{Actual subpaths that $P'$ is split into}
            \begin{enumerate}
                \item For $v \in S$ in decreasing order of $\card{V_v^{\textrm{fringe}} \cap P'}$
                    \begin{enumerate}
                        \item Make $V_v^{\textrm{fringe}} \cap P'$ a subpath.
                            This is a \emph{fringe actual subpath}.
                        \item $P' \gets P' \setminus V_v^{\textrm{fringe}}$
                            \qquad\qquad \textcolor{blue}{\texttt{(this affects the ordering of the loop since $P'$ is changed)}}
                    \end{enumerate}
                \item The remaining $P'$ can be divided into subpaths in core problems naturally, like \JLS.
                    These are the \emph{core actual subpaths}.
            \end{enumerate}
        \end{tbox}
        We have omitted edge cases (e.g. if there is no first/last vertex) here for the sake of clarity, and some justification for why this works, but they can be found in the proof of Lemma 8 in \cite{cao2020efficient}.
        Suffice to say $P'$ is split into $O(\card{S})$ subproblems and thus has $O(\sqrt{k}\card{S})$ children (the logical subproblems).
        To see the last part, note that $\length(P') \le D_r$ and $D_r = O(\sqrt{k}D_{r+1})$.
\end{itemize}

So far nothing new has been said in this paper about \CFR; we have just given some background on their hopset and how the analysis is set up (namely, defining the logical subproblem tree which upper bounds (up to a constant factor) the number of hops $P$ is shortcutted to up to some $(1+\eps)$ error).

We need the following key lemma as a blackbox (this is the analogue of \Cref{prop:jls-2}), and then we are ready to prove our results.

\begin{proposition}[Paraphrasing of Lemma 9 from \cite{cao2020efficient}, adjusted to weighted digraphs]
\label{prop:cfr-3}
    Let $P'$ be a level $r$ logical subproblem in recursive instance $G'$, and let $S = \set{v: \ell(v) = r \textrm{ and } v \in \drelevant{G'}{P'}{\eta_v D_r}}$.
    Let the (possibly empty) level $r+1$ core actual subpaths that $P'$ is split into be $P'_0, P'_1, \ldots, P'_{\card{S}}$.
    Then
    \begin{align*}
        \expect{\sum_{i = 0}^{\card{S}} \card{\drelevant{G'_i}{P'_i}{\eta_{\textrm{min}}D_r}} \;\middle|\; \card{S}} \le \frac{3}{\card{S} + 1} \card{\drelevant{G'}{P'}{\eta_{\textrm{min}}D_r}}.
    \end{align*}
\end{proposition}

The above proposition handles core actual subpaths, but we also need to handle fringe actual subpaths that $P'$ is split into.

\begin{lemma}
\label{lem:fringe-shrink}
    Let $P' = \langle v_0, \ldots, v_h \rangle$ be a level $r$ logical subproblem in recursive instance $G'$.
    Let the level $r+1$ fringe actual subpaths that $P'$ is split into be $P'_1, P'_2, \ldots, P'_Z$.
    Note that $Z$ is random.
    Then
    \begin{align*}
        \expect{\sum_{i = 1}^{Z} \card{\drelevant{G'_i}{P'_i}{\eta_{\textrm{min}}D_r}}} \le 5.
    \end{align*}
\end{lemma}
\begin{proof}
    Assume $\cE_{\ref{prop:cfr-1}}$ holds.\footnote{The contribution from $\cE_{\ref{prop:cfr-1}}$ not holding is $\prob{\textrm{not }\cE_{\ref{prop:cfr-1}}}\expect{\sum_{i = 1}^{Z} \card{\drelevant{G'_i}{P'_i}{\eta_{\textrm{min}}D_r}} \;\middle|\; \textrm{not }\cE_{\ref{prop:cfr-1}}} < 1$. This follows from \Cref{prop:cfr-1} and $\sum_{i = 1}^{Z} \card{\drelevant{G'_i}{P'_i}{\eta_{\textrm{min}}D_r}} \le n^2$.}
    Then $\card{\drelevant{G'}{P'}{\eta_{\textrm{min}}D_r}} \le 2nk^{-r}$ since $\card{\danc{G'}{v_0}{(\eta_{\textrm{min}} +1)D_r}} \le nk^{-r}$ and $\card{\ddesc{G'}{v_h}{(\eta_{\textrm{min}}+1)D_r}} \le nk^{-r}$ and the length of $P'$ is at most $D_r$.
    Therefore, $\expect{Z} \le 2 \lambda k \log n$ since the sampling rate for pivots at level $r$ is $\lambda k^{r+1} \log n / n$.

    Let $E = \set{Z \le 10 \lambda k \log n}$.
    By a Chernoff bound, $\prob{\textrm{not } E} = O(n^{-10})$.
    We can further condition on $E$ since $\prob{\textrm{not } E}\expect{\sum_{i = 1}^{Z} \card{\drelevant{G'_i}{P'_i}{\eta_{\textrm{min}}D_r}} \;\middle|\; \textrm{not }E} < 1$.

    We thus have
    \begin{align*}
        \expect{\sum_{i = 1}^{Z} \card{\drelevant{G'_i}{P'_i}{\eta_{\textrm{min}}D_r}}}
        &
        \le
        \expect{\sum_{i = 1}^{10 \lambda k \log n} \card{\drelevant{G'_i}{P'_i}{\eta_{\textrm{min}}D_r}}}
        + 2
        \tag{Pay $+1$ for each of $\cE_{\ref{prop:cfr-1}}, E$ holding}
        \\
        &
        =
        \paren{\sum_{i = 1}^{10 \lambda k \log n} \expect{\card{\drelevant{G'_i}{P'_i}{\eta_{\textrm{min}}D_r}}}}
        + 2
        \\
        &
        \le
        10 \lambda k \log n \cdot \frac{1}{4 \lambda k \log n}
        + 2
        \tag{\Cref{prop:cfr-2}}
        \\
        &
        \le
        5.
    \end{align*}
\end{proof}

\begin{lemma}[Logical subproblem tree analgoue of \Cref{lem:topdown-cfr}]
\label{lem:topdown-cfr}
    For any $x$, let $X$ be the random variable counting the number of level $r$ nodes $P'$ for all $r$ in the (unpruned) logical subproblem tree such that $\card{\drelevant{G_{P'}}{P'}{\eta_{\textrm{min}}D_r}} \ge x$.
    \begin{align*}
        \expect{X} \le \sqrt{\frac{n^{1 + O(1/\log k)}}{x}}.
    \end{align*}
\end{lemma}
\begin{proof}
    The proof is very similar to that of \Cref{lem:topdown-jls}, only now with the added complication of accounting for secondary copies.
    Let $T$ be the (unpruned) logical subproblem tree.
    \ \\
    \underline{Potential Function}:
    For all $r$, we use the potential function $\phi(P') = \sqrt{\card{\drelevant{G_{P'}}{P'}{\eta_{\textrm{min}}D_r}}}$ for a level $r$ node $P'$, and will show that $\expect{\sum_{v \in T} \phi(P')} \le n^{1/2 + O(1/\log k)}$.
    If we can show this, we are done since nodes $P'$ with $\card{\drelevant{G_{P'}}{P'}{\eta_{\textrm{min}}D_r}} \ge x$ contribute at least $\sqrt{x}$ each to the aforementioned sum; there can thus be at most $X \le n^{1/2 + O(1/\log k)}/\sqrt{x}$ such of them.
    \\
    \underline{Local Step}: We first show
    \begin{align*}
         \expect{\sum_{\substack{P'' \in \children{P'}:\\P'' \textrm{ is a main copy}}} \phi(P'')} \le 7\expect{\phi(P')}.
         \tag{\ding{96}}
    \end{align*}
    The expectations in the following chain of inequalities are conditioned on the value of $\phi(P')$.
    {\allowdisplaybreaks
    \begin{align*}
        &
        \expect{\sum_{\substack{P'' \in \children{P'}:\\P'' \textrm{ is a main copy}}} \phi(P'')}
        =
        \expect{\sum_{\substack{P'' \in \children{P'}:\\P'' \textrm{ is a main copy}\\P'' \textrm{ is fringe}}} \phi(P'')} +
        \expect{\sum_{\substack{P'' \in \children{P'}:\\P'' \textrm{ is a main copy}\\P'' \textrm{ is core}}} \phi(P'')}
        \\
        &
        \le
        5 +
        \expect{\sum_{\substack{P'' \in \children{P'}:\\P'' \textrm{ is a main copy}\\P'' \textrm{ is core}}} \phi(P'')}
        \tag{\Cref{lem:fringe-shrink}}
        \\
        &
        =
        5 +
        \sum_{C \in [n]} \prob{\card{\children{P'}} = C} \expect{\sum_{\substack{P'' \in \children{P'}:\\P'' \textrm{ is a main copy}\\P'' \textrm{ is core}}} \phi(P'') \;\middle|\; \card{\children{P'}} = C}
        \\
        &
        \le
        5 + 
        \sum_{C \in [n]} \prob{\card{\children{P'}} = C} \sqrt{C} \cdot \expect{\sqrt{\sum_{\substack{P'' \in \children{P'}:\\P'' \textrm{ is a main copy}\\P'' \textrm{ is core}}} \card{\drelevant{G_{P''}}{P''}{\eta_{\textrm{min}}D_{r+1}}}} \;\middle|\; \card{\children{P'}} = C}
        \tag{Cauchy-Schwarz}
        \\
        &
        \le
        5 + 
        \sum_{C \in [n]} \prob{\card{\children{P'}} = C} \sqrt{C} \cdot \expect{\sqrt{\sum_{\substack{P'' \in \children{P'}:\\P'' \textrm{ is a main copy}\\P'' \textrm{ is core}}} \card{\drelevant{G_{P''}}{P''}{\eta_{\textrm{min}}D_r}}} \;\middle|\; \card{\children{P'}} = C}
        \tag{$\card{\drelevant{G_{P''}}{P''}{\eta_{\textrm{min}}D_{r+1}}} \le \card{\drelevant{G_{P''}}{P''}{\eta_{\textrm{min}}D_{r}}}$}
        \\
        &
        \le
        5+
        \sum_{C \in [n]} \prob{\card{\children{P'}} = C} \sqrt{C \cdot \expect{\sum_{\substack{P'' \in \children{P'}:\\P'' \textrm{ is a main copy}\\P'' \textrm{ is core}}} \card{\drelevant{G_{P''}}{P''}{\eta_{\textrm{min}}D_r}} \;\middle|\; \card{\children{P'}} = C}}
        \tag{Jensen's Inequality}
        \\
        &
        \le
        5+
        \sum_{C \in [n]} \prob{\card{\children{P'}} = C} \sqrt{C \cdot \frac{3}{C} \card{\drelevant{G'}{P'}{\eta_{\textrm{min}}D_r}}}
        \tag{\Cref{prop:cfr-3}}
        \\
        &
        =
        5 + \sqrt{3} \phi(P') < 7 \phi(P').
        \tag{$1 \le \phi(P')$}
    \end{align*}
    }
    Using $\expect{\expect{X \;\middle|\; \phi(P')}} = \expect{X}$ on both sides gives $\expect{\sum_{P'' \in \children{P'}: P'' \textrm{ is a main copy}} \phi(P'')} \le 7\expect{\phi(P')}$, establishing \ding{96}.
    \\
    \underline{Summing the Pieces Up}:
    We may assume $\cE_{\ref{prop:cfr-1}}$ holds.
    Note that at level $r$ of the logical subproblem tree, there are $O(\lambda^r k^{r/2})$ secondary copies if the guessed distance $D$ is correct.
    The total contribution from level $r$ secondary copies is thus $O(\lambda^r k^{r/2}) \cdot \phi(O(n/k^r)) = O(\lambda^r n^{1/2}) = n^{1/2+O(1/\log k)}$ where the last equality comes from $r = O(\log n / \log\log n)$ and $\lambda$ is a constant in the final parameter setting.
    Since there are $O(\log n / \log\log n)$ levels, the total contribution from all secondary copies is no more than $n^{1/2 + O(1/\log k)}$.
    
    The rest of the proof is then easily completed by summing up over the nodes of $T$ and using the base case $\phi(P) \le \sqrt{n}$.
    {\allowdisplaybreaks
    \begin{align*}
        \expect{\sum_{P' \in T} \phi(P')}
        &
        =
        \expect{\sum_{\substack{P' \in T:\\P' \textrm{ is a secondary copy}}} \phi(P')}
        +
        \expect{\sum_{\substack{P' \in T:\\P' \textrm{ is a main copy}}} \phi(P')}
        \\
        &
        = n^{1/2 + O(1/\log k)}+
        \sum_{r \in [O(\log n / \log k)]} \expect{\sum_{\substack{P' \textrm{ in level } r \textrm{ of } T:\\P' \textrm{ is a main copy}}} \phi(P')}
        \tag{$T$ has $O(\log n / \log k)$ levels}
        \\
        &
        \le n^{1/2 + O(1/\log k)}+
        \sum_{r \in [O(\log n / \log k)]} 7^r \expect{\phi(P)}
        \tag{\ding{96}}
        \\
        &
        \le n^{1/2 + O(1/\log k)}+
        \sqrt{n} \sum_{r \in [O(\log n / \log k)]} 7^r 
        \tag{$\phi(P) \le \sqrt{n}$}
        \\
        &
        =
        n^{1/2 + O(1/\log k)}.
    \end{align*}
    }
\end{proof}

Finally, we show that the number of children any node has in the logical subproblem tree is well controlled.
\begin{observation}
\label{obs:few-children-2}
    With probability at least $1 - O(n^{-10})$, the event $\cE_{\ref{obs:few-children-2}}$ where every node in the logical subproblem tree has $O(k^{3/2} \log n)$ children holds.
\end{observation}
\begin{proof}
    We will condition on $\cE_{\ref{prop:cfr-1}}$ holding, which occurs with probability $1 - O(n^{-10})$ by \Cref{prop:cfr-1}.
    Let $P' = \langle v_0, v_1, \ldots, v_h \rangle$ be any level $r$ node in the logical subproblem tree, contained in some level $r$ recursive $G'$.
    Then $\card{\drelevant{G'}{P'}{\eta_{\textrm{min}}D_r}} \le 2nk^{-r}$ since $\card{\danc{G'}{v_0}{(\eta_{\textrm{min}} +1)D_r}} \le nk^{-r}$ and $\card{\ddesc{G'}{v_h}{(\eta_{\textrm{min}}+1)D_r}} \le nk^{-r}$ and the length of $P'$ is at most $D_r$.

    Let $S = \set{v: \ell(v) = r \textrm{ and } v \in \drelevant{G'}{P'}{\eta_v D_r}}$ and $X_{P'} = \card{S}$.
    Since the sampling rate at level $r$ is $\lambda k^{r+1}\log n / n$, we have $\expect{X_{P'}} \le \lambda k \log n$.
    By a Chernoff bound, $\prob{X_{P'} > 2 k\log n} \le O(n^{-12})$ under the final parameter setting.
    Therefore, $\prob{\bigcup_{\textrm{nodes }P'} \set{X_{P'} > 2 k\log n}} \le O(n^{-10})$.
    This implies that with probability at least $1 - O(n^{-10})$, we have $P'$ being split into at most $O(k \log n)$ actual subpaths, each which can be further split into $O(\sqrt{k})$ logical subpaths.
    That is, $\prob{\cE_{\ref{obs:few-children-2}}} \ge 1 - O(n^{-10})$.
\end{proof}

%% file: prune-2.tex
\subsection{Using \pruneii on \CFR and Analysis of the Hopset}
\label{subsec:prune-2}

Finally, we prove \Cref{thm:seq-hopset} in this subsection.

To begin, \pruneii uses a truncated exact single source shortest path (in both directions) algorithm as a subroutine call.
Truncated Bidirectional SSSP:
Let $Y_{s\rightarrow}$ (resp. $Y_{s\leftarrow}$) be $y$ nearest vertices that $s$ can reach (resp. that reaches $s$) in $G$, breaking ties arbitrarily.
We use $\TDijkstra(G, s, y)$ to denote all edges $st$ where $t \in Y_{s\rightarrow}$ and $ts$ where $t \in Y_{s\leftarrow}$.
Note that $\TDijkstra(G, s, y)$ can be computed in $O(y^2)$ time by doing a truncated Dijkstra's Algorithm in $G$ and in $G$ with edge orientations reversed.

The very simple modification to \CFR is then described below (the text is mostly \CFR from the previous section, with the only substantial change being the addition of the \ding{47} line), where we use the same unmodified outer shell.

\begin{tbox}
    \algname{\CFR \textcolor{\cpruneii}{with \pruneii}}\\
    \textbf{Global Parameters:} $k$, $L$, $\lambda$, and $\rho$ are global parameters to be fixed later. $\ell(v)$ is the recursion level at which $v$ becomes a pivot. $n$ is the number of vertices in the base input graph (it thus remains fixed through all recursive calls).\\
    \textbf{Input:} A directed graph $G = (V, E)$ and a recursion level $r$.\\
    \textbf{Output:} A directed approximate hopset $H \subseteq V^2$.
    \begin{enumerate}
        \item $D_r \gets \frac{D}{\lambda^r k^{r/2}}$.
        \item \ding{47} \textcolor{\cpruneii}{\pruneii: For all $v \in V$, add all edges in $\TDijkstra(G, v, \rho^2)$ to $H$.}
        \item For each $v \in V$ with $\ell(v) = r + L$:
            \begin{itemize}
                \item Add $uv$ with weight $\dist_G(u,v)$ to $H$ for all $u \in \danc{G}{v}{32 \lambda^2 k^2 D_r \log^2 n}$.
                \item Add $vu$ with weight $\dist_G(v,u)$ to $H$ for all $u \in \ddesc{G}{v}{32 \lambda^2 k^2 D_r \log^2 n}$.
            \end{itemize}
        \item For each $v \in V$ with $\ell(v) = r$:
            \begin{enumerate}
                \item Sample $\sigma_v \in [1, 4\lambda^2 k \log^2 n]$ uniformly at random.
                \item \textcolor{blue}{\texttt{(minimizing fringe set size)}}\\
                    Choose $\eta_v \in 4k[(\sigma_v - 1), \sigma_v] + 16 \lambda^2 k^2 \log^2 n$ such that $\card{\drelevant{G}{v}{(\eta_v+1)} \setminus \drelevant{G}{v}{(\eta_v-1)}}$ is minimized.
                \item \textcolor{blue}{\texttt{(fringe set)}}\\
                    $V_v^{\textrm{fringe}} \gets \drelevant{G}{v}{(\eta_v+1)D_r} \setminus \drelevant{G}{v}{(\eta_v-1)D_r}$.
                \item \textcolor{blue}{\texttt{(solve fringe problem)}}\\
                    $H \gets H \cup \CFR \textrm{ \textcolor{\cpruneii}{with \pruneii} on } G[V_v^{\textrm{fringe}}], r+1, D$.
                \item Add $\vanc{v}$ label to all $u \in \danc{G}{v}{\eta_v D_r} \setminus \ddesc{G}{v}{\eta_v D_r}$.
                \item Add $\vdesc{v}$ label to all $u \in \ddesc{G}{v}{\eta_v D_r}  \setminus \danc{G}{v}{\eta_v D_r}$.
                \item Add \ding{55} label to all $u \in \danc{G}{v}{\eta_v D_r} \cap \ddesc{G}{v}{\eta_v D_r}$.
            \end{enumerate}
        \item $V_1, V_2, \ldots, V_t \gets$ Partition of all $v \in V$ with no \ding{55} label, such that $x, y \in V_i$ iff $x$ and $y$ have the exact same labels.
            \qquad\qquad \textcolor{blue}{\texttt{(core subproblems)}}
        \item Output $H \cup \paren{\bigcup_{i \in [t]} \CFR \textrm{ \textcolor{\cpruneii}{with \pruneii} on } G[V_i], r+1, D}$.
            \qquad\qquad \textcolor{blue}{\texttt{(solve core subproblems)}}
    \end{enumerate}
\end{tbox}

We now count the nodes in the pruned logical subproblem tree, which upper bounds the number of hops that $P$ is shortcut to within a multiplicative $(1+\eps)$ error.
Below, we describe how \pruneii allows us to prune the logical subproblem tree.

\begin{ttbox}
    \algname{\pruneii on the Logical Subproblem Tree}\\
    Consider a level $r$ node $P' = \langle v_0, \ldots, v_h \rangle$ in the original logical subproblem tree satisfying $\card{\drelevant{G'}{P'}{\eta_{\textrm{min}}D_r}} \le \rho^2$.
    We remove the entire subtree rooted at $P'$ and keep $P'$ as a leaf.
\end{ttbox}
In the above event, the call to $\TDijkstra(G,v_0,\rho^2)$ will shortcut $P'$ with no error seeing that $v_0 \preceq_{\eta_{\textrm{min}}D_r} v_i$ for all $i \in [h]$.
With this, we are finally ready to prove the main theorems of this section.

\begin{remark}
\label{rmk:hopset}
    The reader may wonder what the obstacles are to using a stronger pruning strategy like \prunei, which computes the Transitive Closure on balls centered around pivots (analogously here, we would want to compute approximate All-Pairs Shortest Paths in said balls which can indeed be done using fast matrix multiplication).
    When a node $P'$ is ``small'', we would then be able to aggressively prune its subtree to the extent that the pruned subtree has size $n^{o(1)}$.
    
    The main obstacle we have faced when trying to adapt a version of \prunei to \CFR is that we are no longer able to upper bound the size of the aforestated subtrees by $n^{o(1)}$ since nodes in the subtree may split into logical subproblems too quickly (i.e. the size of the subtree is lower bounded by the logical subproblem splitting rate).

    We have tried making simple changes to \CFR (e.g. trying to tweak parameters such that the logical subproblem splitting rate is slower), but have not found a way to get around this issue so far.
\end{remark}

\seqhop*
\begin{proof}
    \ \\
    \underline{Hops}:
    Let $P$ be any shortest path in $G$, and $T_P$ its pruned logical subproblem tree.
    The proof is similar to that of \Cref{thm:seq-shortcut-tradeoff-side1}.
    Assume $\cE_{\ref{obs:few-children-2}}$ holds\footnote{The contribution from $\cE_{\ref{obs:few-children-2}}$ not holding is negligible since $\prob{\textrm{not }\cE_{\ref{obs:few-children-2}}}\expect{\card{V(T_P)} \;\middle|\; \textrm{not }\cE_{\ref{obs:few-children-2}}} < 1$. This follows from \Cref{obs:few-children-2} and $\card{V(T_P)} \le n$.}.
    Let $L$ be the number of nodes $P'$ (in subproblem $G'$) over all levels $r$ such that $\card{\drelevant{G'}{P'}{\eta_{\textrm{min}}D_r}} > \rho^2$; we call such nodes larges and otherwise call a node small.
    We have
    \begin{align*}
        \expect{\card{V(T_P)}}
        &
        \le
        \expect{L + \sum_{\textrm{large } P'} \sum_{\substack{P'' \in \children{P'}\\: P'' \textrm{ small}}} \card{V(T_{P''})}}
        \\
        &
        \le
        \expect{L} + \expect{L \cdot \underset{\ref{obs:few-children-2}}{O(\log^{5/2} n)}}
        \tag{\Cref{obs:few-children-2} under the final parameter setting}
        \\
        &
        =
        (1 + O(\log^{5/2} n))\expect{L} = n^{1/2 + o(1)}/\rho.
        \tag{\Cref{lem:topdown-cfr}}
    \end{align*}
    \\
    \underline{Error}:
    The $(1+\eps)$ factor error under the final parameter setting is shown in the unpruned logical subproblem tree by considering for each recursive level $r$ leaf node $P'$ (which is a leaf because a bridge, say $p(P')$, was sampled as a pivot at level $r$), and using the level $r-L$ shortcut $p(P')$ to shortcut the level $r-L$ ancestor node of $P'$.
    A calculation in \cite{cao2020efficient} shows that this gives a $(1+\eps)$ factor error.
    The error in the pruned logical subproblem tree is bounded above by the aforestated error since we no longer count the error accumulated by leaves in the subtree of a pruned node $P'$; the call to $\TDijkstra(G,v_0,\rho^2)$ will shortcut $P'$ with no error.
    \\
    \underline{Time}:
    The time from \CFR under the final parameter settings is $\Ot(m/\eps^2)$ by \Cref{prop:cfr-main} and so it suffices to bound the time from \pruneii calls by $\Ot(n\rho^4)$.
    There are logarithmically many distance guesses and, for each guess, each vertex $v$ makes calls to $\TDijkstra(\cdot,v,\rho^2)$ which take time $\rho^4$ each.
    Summing over each vertex and each of the $O(\log n / \log\log n)$ levels of recursion gives $\Ot(n\rho^4)$.
    \\
    \underline{Size}:
    From \Cref{prop:cfr-main}, \CFR adds $\Ot(n/\eps^2)$ edges.
    \pruneii adds $\Ot(n\rho^2)$ more edges since each call to $\TDijkstra(\cdot,v,\rho^2)$ adds up to $\rho^2$ edges and there are $\Ot(n)$ many such calls.
    \\
    \underline{Main case ($\Ot(m)$ size hopset)}:
    Set $\rho = (m/n)^{1/4}$.
\end{proof}

%% file: sssp.tex
\section{Parallel SSSP}
\label{sec:sssp}

In this short section, we parallelize \CFR \textcolor{\cpruneii}{with \pruneii}, the approximate hopset construction shown in \Cref{sec:seq-hopset}, proving \Cref{thm:par-hopset}.
We then use the approximate hopset to prove \Cref{thm:par-sssp}, our main result for parallel SSSP.

We use the following result from \cite{blackbox} to reduce the construction of hopsets on digraphs to that on so-called shallow digraphs.

\begin{proposition}[Paraphrasing Theorem 3.1 from \cite{blackbox}]
\label{prop:shallow-hopset}
    Suppose $\gamma > c_0 \log^3 n \cdot (1/\eps^2 + 1)$ and $\alpha < 1/(c_0\log^2_\gamma n)$ where $c_0$ is a sufficiently large constant.
    
    Suppose there is a parallel algorithm  $\cA_0$ that, given a digraph $G_{0}$ with $n$ vertices and $m_0$ edges and $(1 + \eps/2)$-approximate shortest path hopbound $h_0 = \gamma \beta$, returns a $(\beta, \eps)$-hopset of size $\alpha m_0 + f(n)$.
    
    Then there is a randomized parallel algorithm $\cA$ that, given a digraph $G$ with $n$ vertices and $m$ edges, returns a $(\beta, \eps \cdot O(\log^2_{\gamma}n))$-hopset of size $S(m) = \Ot(\alpha m + f(n))$.
    $\cA$ makes a polylogarithmic number of sequential calls to $\cA_0$ on digraphs with at most $S(m) + m$ edges, and takes an additional $\Ot(m)$ work and $\Ot(\gamma \beta)$ span.
\end{proposition}

We are now ready to parallelize \CFR \textcolor{\cpruneii}{with \pruneii}.

\begin{restatable}[Parallel Near Linear Work Hopset Construction]{theorem}{parhop}
\label{thm:par-hopset}
    There is an $\Ot(m/\eps^4)$ work
    and $\sqrt[4]{n^{3 + o(1)}/m} \cdot \eps^{-3}$ span
    randomized parallel algorithm that, given a polynomially bounded non-negative integer weighted digraph $G$,
    outputs with high probability a
    $\paren{\sqrt[4]{n^{3 + o(1)}/m}, \eps}$-hopset $H$ with size 
    $\Ot\paren{m/\eps^2}$.
\end{restatable}
\begin{proof}
    \ \\
    \underline{Implementing $\cA_0$ --- Parallel \CFR \textcolor{\cpruneii}{with \pruneii} on shallow digraphs}:
    An issue which arises when parallelizing \CFR \textcolor{\cpruneii}{with \pruneii} is that, while the distance-limited searches use (distance-limited) Dijkstra's algorithm in the sequential setting, this is inherently not parallel.
    We use a similar rounding technique to what \CFR uses to parallelize their algorithm (that is taken from \cite{klein1997randomized}), which takes care of this.
    The changes to our algorithm are only in the outer shell, which we describe later below (new lines marked with a \ding{47}), and with this change also setting $\eps \gets \eps_0/9$.
    For each guessed distance $\Delta = 2^j$, we will run the inner call of \CFR \textcolor{\cpruneii}{with \pruneii} on the graph $G_{\Delta}$, which is the same as $G$ but with edge weights rounded up to integer multiples of $\frac{\eps_0 \Delta}{9h_0}$.
    On $G_{\Delta}$, distance-limited searches in the inner call can be performed by a BFS which treats $\frac{\eps_0 \Delta}{9h_0}$ as one unit.
    \begin{tbox}
        \algname{Modified Outer Shell of \CFR}\\
        \textbf{Global Parameters:} $k$, $L$, $\lambda$, and $c$ are global parameters to be fixed later.\\
        \textbf{Input:} A weighted directed graph $G = (V, E)$ with $(1+\eps_{0}/2)$-approximate shortest path hopbound $h_{0}$.\\
        \textbf{Output:} A weighted directed approximate hopset $H \subseteq V^2$.
        \begin{enumerate}
            \item Repeat $\Theta\paren{\log n}$ times: \qquad\qquad \textcolor{blue}{\texttt{(boosting to succeed whp)}}
                \begin{enumerate}
                    \item For $j \in [-1, \log(nW)]$: \qquad\qquad \textcolor{blue}{\texttt{(guessing a distance $\approx 2^j$)}}
                        \begin{enumerate}
                            \item For each $v \in V$, set $\ell(v) \gets$ the first $r \in [0,\log_k n]$ such that a fresh $\textrm{Bernoulli}\paren{\frac{\lambda k^{r+1}\log n}{n}}$ succeeds.
                            \item \ding{47} $\Delta \gets 2^j$
                            \item \ding{47} $G_{\Delta} \gets G$ with edge weights rounded up to a integer multiples of $\frac{\eps_0 \Delta}{9h_0}$
                            \item For each $v \in V$ with $\ell(v) \le L$:
                                \begin{itemize}
                                    \item Add $uv$ with weight $\dist_{G_\Delta}(u,v)$ to $H$ for all $u \in \danc{G_{\Delta}}{v}{2^{j+1}}$.
                                    \item Add $vu$ with weight $\dist_{G_\Delta}(v,u)$ to $H$ for all $u \in \ddesc{G_{\Delta}}{v}{2^{j+1}}$.
                                \end{itemize}
                            \item $H \gets H\ \cup$ \CFR on $G_{\Delta}, r = 0, D=2^j/k^c$.
                        \end{enumerate}
                \end{enumerate}
            \item Output $H$.
        \end{enumerate}
    \end{tbox}
    The above changes give an algorithm that runs in $\Ot(m/\eps_0^2 + n \rho^4)$ work, and $\Ot(h_0/\eps_0)$ depth since the BFS calls (i.e. depth-limited searches) explore BFS trees only up to layer $\eta_{\textrm{max}}\Delta / \frac{\eps_0 \Delta}{9h_0} = \Ot(h_0/\eps_0)$.
    Let $\beta = n^{1/2+o(1)}/\rho$.
    To see that the algorithm returns a $(\beta, \eps_0)$-hopset, we analyze the logical subproblem tree of an arbitrary path $P$ which has hopbound $h_0$ and $(1+\eps_{0}/2)$-approximates the shortest path $P^*$ between its endpoints; say $P$ has length in $[\Delta,2\Delta]$.
    Let us subscript paths with $\Delta$ when their lengths should be considered under $G_{\Delta}$.
    In $G_{\Delta}$, then, $P_{\Delta}$ will have length no more than $(1+\eps_{0}/9)$ times $P$ and be shortcut with high probability to some $\beta$ hop path $\bar{P}_{\Delta}$ that has length at most $(1+\eps_{0}/9)$ times longer than $P_{\Delta}$.
    That is,
    \begin{align*}
        \bar{P} \le \bar{P}_{\Delta} \le (1+\eps) P_{\Delta} = (1+\eps_{0}/9) P_{\Delta} \le (1+\eps_{0}/9)^2 P \le (1+\eps_{0}/2)(1+\eps_{0}/9)^2 P^* \le (1+\eps_{0})P^*
    \end{align*}
    where we have abused notation and used the symbols for paths to refer to their lengths.
    \\
    \underline{Setting parameters}:
    Fix $\rho = (m_0/n)^{1/4}/\textrm{polylog}(n)$ for a sufficiently large power in the polylogarithm term.
    Let the output of \CFR \textcolor{\cpruneii}{with \pruneii} on shallow digraphs be a $(\beta,\eps/\Theta(\log^2_{\gamma}n))$-hopset $H$ with $\beta = n^{1/2 + o(1)}/\rho = \sqrt[4]{n^{3 + o(1)}/m_0}$.
    Note that for this setting of $\rho$, we have $\card{H} = \Ot(n/\eps^2 + n\rho^2) = \alpha m_0 + f(n)$ where $\alpha < 1/(c_0\log^2_\gamma n)$ and $f(n) = \Ot(n/\eps^2)$, whence $S(m) = \Ot(m/\eps^2)$.
    \\
    \underline{Putting things together}:
    Invoking \Cref{prop:shallow-hopset}, $\cA$ outputs with high probability a
    $(\sqrt[4]{n^{3 + o(1)}/m},\eps)$-hopset with size $\Ot\paren{m/\eps^2}$
    in $\Ot(m/\eps^4)$ work and $\sqrt[4]{n^{3 + o(1)}/m} \cdot \eps^{-3}$ span.
\end{proof}

Finally, we proceed to proving \Cref{thm:par-sssp}, our result for SSSP, using the reduction below.

\begin{proposition}[Paraphrasing and composing Theorem 1.7 with Lemma 3.1 from \cite{rozhovn2023parallel}]
\label{prop:reduction}
    There is an algorithm that ---
    given a digraph $G=(V,E)$ with polynomially bounded non-negative integer weights and a source $s\in V$ ---
    computes $\dist(s,v)$ for all $v \in V$, using $\textrm{polylog}(n)$ calls to a $(1 + O(1/ \log n))$-approximate distance oracle $\mathcal{O}$ on digraphs and an additional $\Ot(m)$ work and $\textrm{polylog}(n)$ depth.\footnote{One might notice that Lemma 4.1 in \cite{klein1997randomized} is also about reducing exact shortcut path to approximate shortest path oracles. However, they require a stronger approximate shortest path oracle: The approximate distance must be an exact distance on the shortcut $G'$. Thus, it cannot be used here.}
\end{proposition}

\parsssp*
\begin{proof}
    Once $d(s,v)$ has been computed for all $v \in V$, the shortest path tree rooted at $s$ can be computed in $\Ot(m)$ work and $\textrm{polylog}(n)$ depth (by doing a search in the neighborhood of each vertex for its parent), so it remains to show how to compute $d(s,v)$.
    
    We compute $d(s,v)$ through the reduction of \Cref{prop:reduction}, using the oracle $\mathcal{O}$ whose implementation is described as follows.
    Let $\eps \gets O(1/\log n)$.
    Use the algorithm of \Cref{thm:par-hopset} on $G$ to get a $(\sqrt[4]{n^{3 + o(1)}/m}, \eps)$-hopset $H$ with $\card{H} = \Ot(m)$, and then use the algorithm of Klein and Subramanian~\cite{klein1997randomized} to compute shortest $\sqrt[4]{n^{3 + o(1)}/m}$ hop paths in $G \cup H$ from $s$ in $\Ot(m)$ work and $\sqrt[4]{n^{3 + o(1)}/m}$ depth.
\end{proof}

%% file: open-problems.tex
\section{Open Problems}
\label{sec:open-problems}

There are two main open problems we leave the reader.

\paragraph{Faster Sequential Construction of the Tunable Folklore Hopset. }
Assuming $\omega = 2+ o(1)$, the shortcut set of \Cref{thm:seq-shortcut} matches the folklore $O(n/\sqrt{m})$-shortcut set in its parameters, up to subpolynomial factors.
The punchline here is that for $\omega = 2$, we give a near-linear time construction of the folklore shortcut set (almost).
On the other hand, the $(O(n/\sqrt{m}), \eps)$-\emph{hopset} given by the folklore result is better than the hopset of \Cref{thm:seq-hopset} in its parameters, especially when $m \gg n$.
For example, in the dense regime where $m = \Omega(n^2)$, the folklore construction gives an $(O(1), \eps)$-hopset whereas we give a near-linear time construction of a $(n^{1/4+o(1)},\eps)$-hopset.
For $\omega = 2$, can we get a near-linear time construction of a $(n^{1+o(1)}/\sqrt{m}, \eps)$-hopset, almost matching the folklore result?

\paragraph{Faster Algorithms in the Sparse Regime. }
Our results do not make any improvements when $m = O(n)$; the best near work-efficient algorithm still has depth $n^{1/2+o(1)}$ in this case, for both reachability and SSSP.
It is particularly interesting to get faster algorithms here.

One potential opportunity --- using the same approach of constructing a shortcut set / hopset and then searching --- is to leverage ideas from the breakthrough of \cite{kogan2022new} who showed that for all digraphs $G$, there are $O(n^{1/3})$-shortcut sets $H$ with $\card{H} = \Ot(n)$.
While there are non-trivial constructions of said shortcut sets (see \cite{kogan2022beating,kogan2023faster}), near-linear time sequential algorithms are not known.
The result for shortcut sets was followed up by \cite{bernstein2023closing}, who showed that there are $(O(n^{1/3}), \eps)$-hopsets $H$ with $\card{H} = \Ot(n)$.
A faster sequential construction of these could provide a good first step towards getting faster parallel reachability and SSSP algorithms on sparse graphs.

One should note, however, that using the shortcut set (resp. hopset) framework for reachability (resp. SSSP) is inherently limited: \cite{bodwin2023folklore} showed that there are sparse graphs for which for any $c > 0$ any $n^{1/4 - c}$-shortcut set requires a strongly super-linear size, meaning there is no nearly work-efficient algorithm with $n^{1/4 - c}$ span using precisely this approach.
More generally, \cite{hoppenworth2024new} showed an $n^{2/9}$ barrier for $O(m)$ size shortcut sets.

%% file: sketch-key-lemma.tex
\section{Proof Sketch of \Cref{prop:jls-2}}
\label{app:sketch-key-lemma}

Here, we give a more conceptual sketch for why \Cref{prop:jls-2} should be true.
For a rigorous proof, see Lemma 4.4 in \cite{jambulapati2019parallel}.
\begin{proof}[Proof Sketch]
    For simplicity, assume $P'$ has no bridges.
    Suppose also that all $p \in S$ are ancestors.
    Under these assumptions, we will sketch a proof of the inequality
    \begin{align*}
        \expect{\sum_{i = 1}^{\card{S} + 1} \card{\anc{G'_i}{P'_i}} \;\middle|\; \card{S}} \lesssim \frac{1}{\card{S} + 1} \card{\anc{G'}{P'}}.
    \end{align*}
    To see this, observe that for any two distinct $p,q \in \anc{G'}{P'}$, there is some $x \in \set{p,q}$ such that $x \in S \implies y \not\in \bigcup_i \anc{G'_i}{P'_i}$, where $x \neq y \in \set{p,q}$.
    Indeed if $p$ reaches $q$, then $q \in S$ implies $p$ receives the label ``I reach $q$'' whereas no vertex in $P'_i$ for all $i$ receives such a label as $q$ is an ancestor of $P'$ which implies $p \not\in \bigcup_i \anc{G'_i}{P'_i}$.
    On the other hand if $p$ and $q$ are unrelated (and say $p$ reaches a weakly earlier vertex of $P'$ than $q$ does), then $p \in S$ implies $q$ receives the label ``I am unrelated to $p$'' whereas every vertex that $q$ reaches on $P'$ receives the label ``$p$ reaches me'' which implies $q \not\in \bigcup_i \anc{G'_i}{P'_i}$.
    We can thus form a tournament $T$ on vertex set $\anc{G'}{P'}$ and reframe the above inequality with the purely combinatorial statement on tournaments
    \begin{align*}
        \expect{\# \ v \in V(T)\textrm{ not dominated by } S \;\middle|\; \card{S}} \lesssim \frac{1}{\card{S} + 1} \card{V(T)},
    \end{align*}
    which is true for any tournament $T$.
    
    An analogous inequality holds for when all $p \in S$ are descendants.
    By ``averaging'' these two inequalities, \Cref{prop:jls-2} follows.
\end{proof}